\documentclass[letterpaper,11pt]{article}
\usepackage{hyperref}

\usepackage{mathtools,amsfonts,amsmath,amssymb,latexsym,amsthm,hyperref,cleveref, url}
\usepackage{enumitem}
\usepackage{makecell}
\usepackage{pdflscape}

\newtheorem{ptheorem}{Theorem}[section]
\newtheorem{plemma}[ptheorem]{Lemma}
\newtheorem{pclaim}[ptheorem]{Claim}

\AddToHook{env/plemma/begin}{\crefalias{ptheorem}{plemma}}
\AddToHook{env/pclaim/begin}{\crefalias{ptheorem}{pclaim}}
\AddToHook{env/pcorollary/begin}{\crefalias{ptheorem}{pcorollary}}

\crefname{ptheorem}{Theorem}{Theorems}
\crefname{plemma}{Lemma}{Lemmas}
\crefname{pclaim}{Claim}{Claims}
\crefname{section}{Section}{Sections}
\crefname{table}{Table}{Tables}
\crefname{pcorollary}{Corollary}{Corollaries}

\usepackage{fullpage}

\DeclareMathOperator{\polylog}{polylog}

\usepackage{algcompatible}
\usepackage{algorithm}
\usepackage{algpseudocode}
\usepackage{tabto}

\usepackage{tikz} 
\usetikzlibrary{automata,positioning}

\usepackage{graphicx}

\newcommand\adhconn{\textsf{adhconn}}
\newcommand\bag{\textsf{bag}}
\newcommand\bgraph{\textsf{bgraph}}
\newcommand\pset{\textsf{pset}}
\newcommand\lca{\textsf{lca}}
\newcommand\dir{\textsf{dir}}
\newcommand\torso{\textsf{torso}}
\newcommand\profile{\textsf{profile}}
\newcommand\adh{\textsf{adh}}
\newcommand\mrg{\textsf{mrg}}
\newcommand\comp{\textsf{comp}}
\newcommand\cone{\textsf{cone}}
\newcommand\parent{\textsf{parent}}
\newcommand\children{\textsf{children}}
\newcommand\col{\text{col}}
\newcommand\invcol{\text{invcol}}
\newcommand\sfP{\mathsf{P}}
\newcommand\sfS{\mathsf{S}}
\newcommand\sfU{\mathsf{U}}
\newcommand\sfQ{\mathsf{Q}}
\newcommand\backedge{\textsf{backedge}}
\newcommand\backvertex{\textsf{backvertex}}
\newcommand\parr{\textsf{par}}
\newcommand\depth{\textsf{depth}}

\newcommand\dfs{\textsf{dfs}}

\def\final{0}  % set this to 1 to get a comment-free version

\ifnum\final=0  %namely if we allow comments in the output
\newcommand{\todo}[1]{{\color{red}[{\small TODO: \bf #1}]\marginpar{\color{red}*}}}
\newcommand{\thatchaphol}[1]{{\color{purple}[{\small Thatchaphol: \bf #1}]\marginpar{\color{purple}*}}}
\newcommand{\benyu}[1]{{\color{blue}[{\small Benyu: \bf #1}]\marginpar{\color{blue}*}}}
\else % in this case [final=1] we don't want any comments to show
\newcommand{\thatchaphol}[1]{}
\newcommand{\benyu}[1]{}
\newcommand{\todo}[1]{}
\fi

\global\long\def\poly{\mathrm{poly}}

\title{Connectivity Oracle Under Vertex Failures \\by Shortcutting Unbreakable Decomposition}
\author{
Xizhe Li\thanks{University of Michigan, \texttt{xizheli@umich.edu}.} \and 
Yaowei Long\thanks{University of Michigan, \texttt{yaoweil@umich.edu}.} \and 
David Pidugu\thanks{University of Michigan, \texttt{davidbp@umich.edu}.} \and
Thatchaphol Saranurak\thanks{University of Michigan, \texttt{thsa@umich.edu}. Supported by NSF Grant CCF-2238138 and a Sloan Fellowship.} \and
Benyu Wang\thanks{University of Michigan, \texttt{benyuw@umich.edu}.}
}
\date{}

\begin{document}
\maketitle

\begin{abstract}
We give an improved \emph{connectivity oracle under vertex failures}. After a set of $k$ vertices fails, our oracle performs an $O(k^{6})$-time update \emph{independent of the graph size $n$}, and then answers pairwise connectivity queries in optimal $O(k)$ time. For constant $k$, it uses near-linear space and can be built in near-linear preprocessing time.

In contrast, all prior oracles with $n$-independent update time~\cite{pilipczuk2021algorithms, van2019sensitive} either require $\Omega(n^{2})$ space or incur $2^{2^{O(k)}}$ update and query time. Moreover, their preprocessing time is polynomially large in $n$, far from near-linear.

Our oracle builds on the unbreakable decomposition framework of~\cite{pilipczuk2021algorithms}, but introduces three new ingredients: (i) \emph{shortcutting} over the tree decomposition to reduce space from quadratic to near-linear, (ii) \emph{bootstrapping} that leverages $n$-dependent oracles internally to obtain near-linear preprocessing, and (iii) a new \emph{patch set} mechanism that yields conditionally optimal $O(k)$ query time.
\end{abstract}

\pagenumbering{gobble}
\clearpage
\tableofcontents
\clearpage

\pagenumbering{arabic}

\newpage

\section{Introduction}\label{sec:intro}
We study \emph{connectivity oracles under vertex failures}. Given an undirected graph $G$, the goal is to preprocess $G$ so that after a \emph{batch} of at most $k$ vertices fails, we can quickly answer connectivity queries in the remaining graph $G\setminus S$. This is a fundamental graph data structure in the \emph{dynamic subgraph model}~\cite{frigioni1997dynamically}, and it is also useful as a subroutine for other basic connectivity tasks, such as vertex-cut oracles~\cite{jiang2026new,DBLP:conf/icalp/Kosinas25} and listing shredders~\cite{hua2024finding}.

A particularly natural regime---both in applications and as an algorithmic benchmark---is when the number of failures $k$ is a small constant in practice while the underlying graph is massive. In this setting, even a $\polylog(n)$ overhead per batch can dominate; this motivates designing oracles whose \emph{update} and \emph{query} costs depend only on $k$ (i.e., are independent of $n$).

\paragraph{Problem.}
The oracle supports three phases:
\begin{itemize}
\item \textbf{Preprocess:} Given $G=(V,E)$ with $|V|=n$, $|E|=m$, and a parameter $k$, build a data structure.
The \emph{preprocessing time} is the time to build the structure; the \emph{space} is its memory usage.
\item \textbf{Update:} Given a failure set $S\subseteq V$ with $|S|\le k$, update the data structure in \emph{update time}.
\item \textbf{Query:} Given $u,v\in V$, decide whether $u$ and $v$ are connected in $G\setminus S$ in \emph{query time}.
\end{itemize}
Typically, one update (for a batch $S$) is followed by many connectivity queries under the same failures.

\paragraph{Well-understood \emph{modulo $n^{o(1)}$ factors}.}
There is a long line of work optimizing the tradeoffs among space, preprocessing time, update time, and query time.
SPQR trees~\cite{hopcroft1973dividing,gutwenger2000linear} yield an optimal oracle that handles $k\le 2$ failures using $O(n)$ space, $O(m)$ preprocessing, and $O(1)$ update and query time. Very recently, Kosinas~\cite{DBLP:conf/icalp/Kosinas25} extended this optimal oracle to $k=3$.

For general $k$, the first structures are due to Duan and Pettie~\cite{duan2010connectivity,DuanP20}.
More recently, Long and Saranurak~\cite{long2024} gave a deterministic oracle whose parameters are simultaneously optimal up to $n^{o(1)}$ factors under fine-grained complexity hypotheses:
$O(\min\{m,nk\}\log^{*}n)$ space,
$m^{1+o(1)} + \tilde{O}(mk)$ preprocessing time\footnote{Throughout the paper, we use $\tilde{O}(\cdot)$ to hide polylog factors.},
$O(k^{2} n^{o(1)})$ update time,
and $O(k)$ query time.
Long and Wang~\cite{long2024better} later improved the $n^{o(1)}$ factor in the update time to polylogarithmic.
Although this landscape is essentially settled when mild $n$-dependence is allowed, these results share a common drawback: their update time still contains an $n$-dependent term (subpolynomial or polylogarithmic).

\paragraph{The frontier: eliminating $n$-dependence.}
Existing results essentially fall into two camps.
On the one hand, we have near-linear-space oracles with excellent $k$-dependence, but their update time still depends on $n$.
On the other hand, only two known approaches achieve update and query times depending solely on $k$:
\begin{itemize}
\item The randomized oracle of van den Brand and Saranurak~\cite{van2019sensitive} uses algebraic techniques to obtain $O(k^\omega)$ update time and $O(k^2)$ query time (where $\omega$ is the matrix multiplication exponent), but requires $O(n^2)$ space and $O(n^\omega)$ preprocessing time.
\item The deterministic approach of Pilipczuk et al.~\cite{pilipczuk2021algorithms} is based on unbreakable decomposition.
It achieves $n$-independent update time, but their oracle with linear space comes with doubly exponential dependence on $k$ (and very large preprocessing). They also obtained an alternative oracle with  $\text{poly}(k)$ update and query time, but it is strictly worse than the oracle by \cite{van2019sensitive}  except being deterministic.
\end{itemize}
Thus, prior $n$-independent oracles either use $\Omega(n^2)$ space or incur doubly exponential  dependency on $k$. Moreover, the preprocessing time is far from linear.
See \Cref{tbl:1}.

\paragraph{Our results.}
We resolve both drawbacks: we design a deterministic vertex-failure connectivity oracle with \emph{near-linear space} in $n$ (with only a $\poly(k)$ factor) and \emph{purely $\poly(k)$} update and query time, independent of $n$. Moreover, our preprocessing time is near-linear when $k=O(1)$.

\begin{ptheorem}[Main theorem: near-linear space, $n$-independent update/query]\label{thm:linear}
There exists a deterministic vertex-failure connectivity oracle that uses $O(k^2 n\alpha_c(n))$ space,
$k^{O(k^2)}n + O(m + k^3 n \log^2 n + k^6 n \log n)$ preprocessing time,
$O(k^4)$ update time, and $O(k^3)$ query time.
\end{ptheorem}

Here $\alpha_c(n)$ denotes the extremely slowly growing inverse Ackermann function for any fixed constant $c$.
For constant $k$, the space and preprocessing are near-linear, while update and query are constant-time.

The ideal dependence on $k$ would be $O(k^2)$ update time and $O(k)$ query time: this would be both $n$-independent and conditionally optimal under fine-grained conjectures~\cite{HKNS15,long2024}.
We take a step toward this frontier by showing that optimal $O(k)$ query time \emph{is} achievable with purely $k$-dependent update time, via a tradeoff that allows an additional $2^{O(k^2)}n$ term in space (while keeping preprocessing near-linear for constant $k$).

\begin{ptheorem}[Optimal query time via extra space]\label{thm:main}
There exists a deterministic vertex-failure connectivity oracle that uses $2^{O(k^2)}n + O(k^2 n\alpha_c(n))$ space,
$k^{O(k^2)}n + O(m + k^3 n \log^2 n + k^6 n \log n)$ preprocessing time,
$O(k^6)$ update time, and $O(k)$ query time.
\end{ptheorem}

The $O(k)$ query time in Theorem~\ref{thm:main} is conditionally optimal~\cite{HKNS15}.
Table~\ref{tbl:1} summarizes known vertex-failure connectivity oracles whose \emph{update times are independent of $n$} and compares them with our results.

\renewcommand{\arraystretch}{1.5}
\begin{table}[t]
\centering
\hspace*{-0.02\textwidth}
\small
\setlength{\tabcolsep}{4pt}
\begin{tabular}{|l|c|c|c|c|c|}
\hline
 & Det./Rand. & Space & Preprocessing & Update & Query \\
\hline
\hline

\makecell[l]{Kosinas \\($k \le 3$ only)\\ \cite{DBLP:conf/icalp/Kosinas25}}
& Det.
& $O(n)$
& $O(m)$
& $O(1)$
& $O(1)$ \\
\hline

\makecell[l]{Brand \& Saranurak\\ \cite{van2019sensitive}}
& Rand.
& $O(n^2)$
& $O(n^\omega)$
& $O(k^\omega)$
& $O(k^2)$ \\
\hline

\makecell[l]{Pilipczuk et al.\\ \cite{pilipczuk2021algorithms}}
& Det.
& $n \cdot 2^{2^{O(k)}}$
& $mn^2 \cdot 2^{2^{O(k)}}$
& $2^{2^{O(k)}}$
& $2^{2^{O(k)}}$ \\
\cline{2-6}

& Det.
& $n^2 \mathrm{poly}(k)$
& $\mathrm{poly}(n)\,2^{O(k \log k)}$
& $\mathrm{poly}(k)$
& $\mathrm{poly}(k)$ \\
\hline

\makecell[l]{\textbf{This Paper}}
& Det.
& $O(k^2 n\alpha_c(n))$
& $k^{O(k^2)}n + \widetilde{O}(km + k^6 n)$
& $O(k^4)$
& $O(k^3)$ \\
\cline{2-6}

& Det.
& $2^{O(k^2)}n + O(k^2 n\alpha_c(n))$
& $k^{O(k^2)}n + \widetilde{O}(km + k^6 n)$
& $O(k^6)$
& $O(k)$ \\
\hline
\end{tabular}
\caption{Complexities of known vertex-failure connectivity oracles whose \emph{update times are independent of $n$}. We use $\widetilde{O}(\cdot)$ to hide $\polylog(n)$ factors. Here $2 \le \omega < 3$ is the matrix multiplication exponent, and $\alpha_c(n)$ is the inverse Ackermann function (for any constant $c$).}
\label{tbl:1}
\end{table}
\renewcommand{\arraystretch}{1}

\paragraph{Other consequences: vertex-cut oracles.}
Beyond connectivity queries, our oracle improves other basic connectivity tasks.
In \Cref{sec:cutoracle}, we discuss a simple reduction (refining the DFS-based framework of~\cite{DBLP:conf/icalp/Kosinas25}) that converts any vertex-failure connectivity oracle into a vertex-cut oracle\footnote{A vertex-cut oracle supports vertex-cut query: will deleting a given set of vertices disconnect the graph? See \Cref{sec:cutoracle} for a formal problem definition.} (Theorem~\ref{thm:VertexCutOracle}). More interestingly, this reduction is nearly lossless, and thus it immediately opens up new trade-offs with $n$-independent query time for vertex-cut oracles (the current vertex-cut oracles \cite{jiang2026new,DBLP:conf/icalp/Kosinas25} all have their query times depending on $n$). For example, by plugging in our new $n$-independent connectivity oracle, it gives the following new vertex-cut oracle.

\begin{ptheorem}\label{thm:vertex-cut-oracle:intro}
There exists a deterministic vertex-cut oracle that uses $O(k^{2}n\alpha_{c}(n))$ space,
$k^{O(k^{2})}n+O(m+k^{3}n\log^{2}n+k^{6}n\log n)$ preprocessing time,
and $O(k^{4}+2^{k}k^{2})$ query time.
\end{ptheorem}

Furthermore, we show a similar reduction (\Cref{thm:SteinerVertexCutOracle}) for the more generalized \textit{Steiner vertex-cut oracles} problem (see \Cref{sec:cutoracle} for precise problem definition). Previously, there were results \cite{bhanja2025near} for the related problem of $k$-vertex-fault-tolerant Steiner connectivity labeling schemes, but there are still no explicit results\footnote{\cite{bhanja2025near} states that \cite{jiang2026new}'s vertex-cut oracle ``seems plausible to
extend also to the Steiner variant'', but there is no explicit proof, and the reduction following \cite{DBLP:conf/icalp/Kosinas25}'s DFS-based framework will likely be simpler.} for Steiner vertex-cut oracles. The new reduction immediately yields many new trade-offs for Steiner vertex-cut oracles, leveraging the rich body of existing results for vertex-failure connectivity oracles. However, compared with the first reduction, it incurs larger losses (one more $O(\log n)$ overhead on the query time), and therefore cannot yield reasonable trade-offs with $n$-independent query time.

\paragraph{Techniques.}
We improve upon the oracle of \cite{pilipczuk2021algorithms} with $\Omega(n^2)$ space using three key ideas. 
First, to improve the space to near-linear, we introduce the usage of \emph{shortcut} on top of the tree that represents the unbreakable decomposition of the graph.
Second, to improve the preprocessing time, we observe that we can \emph{bootstrap} the construction of our $n$-independent-update oracle using $n$-dependent-update oracles (by e.g. \cite{kosinas2023,long2024better}). 
Third, to obtain the optimal query time, we introduce the \emph{patch sets}. 
However, applying these three ideas at the same time is non-trivial.

\paragraph{Organization.}
\Cref{sec:overview} gives a high-level explanation of our results.
Section~\ref{sec:prelims} introduces preliminaries and tools.
Section~\ref{sec:prep} describes the preprocessing and data structure.
Section~\ref{sec:update-query} gives the update and query algorithms and proves Theorem~\ref{thm:main}.
Section~\ref{sec:tradeoffs} presents additional configurations and proves Theorem~\ref{thm:linear}.
Finally, Section~\ref{sec:discussions} discusses open problems and future improvements.

% \input{intro_2}
% \section{Introduction}\label{sec:intro}

% \input{intro}

\section{Overview}\label{sec:overview}
\paragraph{The unbreakable decomposition framework.}

We start by introducing the framework inspired by \cite{pilipczuk2021algorithms} using unbreakable decomposition \cite{korhonen2024}. In a graph $G$, we say a set $B \subseteq V$ is $(q,k)$-unbreakable if for any subset $S \subseteq V(G)$ of at most $k$ vertices, there can be at most one connected component $C$ in $G \setminus S$ such that $|C \cap B| > q$ while $|B \setminus C| \leq q$. A $(q, k)$-\textit{unbreakable decomposition} of a graph $G$ is a tree decomposition $\mathcal{T} = (T, \bag)$ where the bag associated with each tree node $x \in V(T)$ is $(q, k)$-unbreakable in $G$. Additionally, we say that $\mathcal{T}$ is \textit{strongly} $(q, k)$-unbreakable if every $\bag(x)$ is $(q, k)$-unbreakable in the subgraph of $G$ induced by the union of all bags in the subtree of $T$ rooted at $x$. Throughout the paper, we will use a strongly $(k, k)$-unbreakable decomposition of adhesion $k$ given by \cite{korhonen2024}, thus in this paper we always assume $q=k$. 

Given the set of at most $k$ failed vertices $S$ and for each $u \in S$, we first identify the topmost node $x \in T$ where its bag contains $u$ and call the node $x$ \emph{important}. Let $X$ be the set of all important nodes and let $Y$ be the least common ancestor closure of $X$ in $T$, adding the root if it is not already included. Note that $|X| \leq k$ and $|Y| \leq O(k)$. For each $x \in X$, we aim to compute the \emph{profile} of $x$, which is a graph on the adhesion of $x$ where two vertices are adjacent if and only if they are connected in the subgraph of $G \setminus S$ induced by the union of all bags in the subtree rooted at $x$. We build these profiles bottom-up, starting from the lowest nodes in $Y$ and working all the way to the root. We use additional precomputed graphs called \emph{torsos} to quickly jump between important nodes and skipping sections of the graph where no vertex failures occur. For any $x, y \in V(T)$ where $x$ is a strict ancestor of $y$, the graph $\torso(x, y)$ is on the adhesions of $x$ and $y$. Two vertices in $\torso(x, y)$ are adjacent if and only if there exists a path between them in $G$ whose interior lies entirely within the the induced subgraph on the union of all bags in the subtree rooted at $x$ and stopping at $y$.

Computing torsos is a main bottleneck in the preprocessing time from \cite{pilipczuk2021algorithms}. To speed up this, we observe that torso queries can be computed via standard connectivity queries with vertex failures. Therefore, we build another oracle using \cite{kosinas2023} to boost the preprocessing time.

To evaluate the query $u, v$, we first consider $u$ and repeat the same process for $v$. We first \emph{augment} $u$ to the computed profile of its closest ancestor in $Y$ while maintaining its connectivity invariants. Then we identify all important nodes along the path from $u$ to the root, and working upward by augmenting $u$ to each of their profiles until the root. After this is repeated for $v$, we can simply check the profile of the root to determine if $u$ and $v$ are connected in $G \setminus S$.

\paragraph{Shortcuts.}

To avoid quadratic space used in storing all torsos \cite{pilipczuk2021algorithms}, we develop an approach using \emph{shortcuts}. We observe that, suppose we have $x,y,z \in V(T)$ where $x$ is an ancestor of $y$ and $y$ is an ancestor of $z$ in $T$, then $\torso(x,z)$ can be computed through $\torso(x,y)$ and $\torso(y,z)$ efficiently.

With this observation, we only need to store a subset of torsos in our oracle. We model this by a set of \emph{shortcut edges} in $T$ that when $x$ is an ancestor of $z$, $x$ can reach $z$ through constant shortcut edges. Therefore, when $\torso(x,z)$ is needed, we just use the stored torsos on those shortcut edges to compute $\torso(x,z)$. From \cite{bodlaender1994,le2025}, we only need $n\alpha_c(n)$ shortcut edges, and thus we only need to store $n\alpha_c(n)$ (instead of $n^2$) torsos.

\paragraph{Patch sets.}

Next, we introduce how our oracle can achieve $O(k)$ optimal query time instead of $\text{poly}(k)$ as in \cite{pilipczuk2021algorithms} with storing some additional information about \emph{patch sets}. As standard local search in \cite{pilipczuk2021algorithms} still has $O(k^2)$ query time even in $(q,k)$-unbreakable graphs, it would be good to assume $G$ is $(q,k)$-unbreakable below while still capturing the hardness and our ideas.

We define a collection $\mathcal{P}$ of \emph{patch sets} as follows. Start from an arbitrary vertex $v \in V$, we perform BFS in $V$ from $v$ until a set $P \subseteq V$ with $(k+1)$ vertices are visited. We add $P$ to $\mathcal{P}$, delete $P$ from $V$, and then repeat the procedure until there's no connected component in $V$ with at least $(k+1)$ vertices. Finally, for each remaining component $P$, we add $P$ into $\mathcal{P}$. We say $P \in \mathcal{P}$ is \emph{large} if $|P|=k+1$ and \emph{small} otherwise. A small patch set can only link to large patch sets.

We know that, for any failed vertex set $S$ with size at most $k$, there exists at most one large component $C$ in $G \setminus S$ with $|C|>k$. We say a patch set $P \in \mathcal{P}$ is \emph{touched} if $P \cap S \neq \emptyset$, and \emph{untouched} otherwise. For each vertex $v$ that's in an touched patch $P$, we do BFS from $v$ in $G \setminus S$. By visiting at most $k+1$ vertices in $O(k^2)$ time, we either find all vertices connected with $v$ in $G \setminus S$, or certify $v \in C$. Since there are at most $k$ affected components, the total update time will be $O(k \cdot k \cdot k^2)=O(k^4)$.

Then consider any vertex $u\in V$. If $u$ is in an affected patch set $P$, we already know the component of $u$. If $u$ is in an unaffected large patch $P$, we also know $u \in C$. Otherwise, $u$ is in an unaffected small patch $P$. By enumerating at most $(k+1)$ neighbors of $P$, if $P$ is not a component itself, we will find some $u'$ that's in a large patch set and connected with $P$. We can then compute the component of $u$ by $u'$. Therefore, given a query $u,v$, we can check if they are in the same connected component using $O(k)$ time.

%\benyu{Please feel free to edit this section if I'm not accurate.}

\section{Preliminaries}\label{sec:prelims}
In this section we introduce and define most of the concepts and notations used in later sections. Many of the definitions and notations used by this paper will follow the convention set by \cite{pilipczuk2021algorithms}.

\paragraph{Graphs.} We will use the standard notion of graphs. Throughout this paper, we will only consider undirected graphs. For any graph $G$, we use $V(G)$ to denote its vertex set and $E(G)$ to denote its edge set. Additionally, we use $n = |V(G)|$ and $m = |E(G)|$ to represent the size of the graph $G$.  WLOG we can assume $m=O(nk)$ by a sparsification in \cite{nagamochi1992computing} with $O(m)$ time. For any subset of vertices $S \subseteq V(G)$, we use $G[S]$ to denote the subgraph of $G$ induced by $S$, and we use $G \setminus S$ to denote the graph $G$ after removing the set of vertices $S$. We will often use the term \textit{BFS in $G$} as a short hand as \textit{breadth-first search} in a graph $G$.

\paragraph{Trees.} A \textit{rooted tree} $T$ is a connected acyclic graph with a distinguished root $r$. For any node $x \in V(T)$, we use $\parent(x)$ to represent the parent of $x$ and $\children(x)$ to represent the set of children of $x$. Let $x, y \in V(T)$ be two nodes, we say that $x$ is an ancestor of $y$ and $y$ is a descendant of $x$ if $x$ is on the unique path from $y$ to $r$. Additionally, we say that $x$ is a strict ancestor of $y$ if $x$ is an ancestor of $y$ and $x \neq y$. For a node $x \in V(T)$, we let $T(x)$ represent the subtree of $T$ rooted at $x$ that contains all descendants of $x$. Also, for two nodes $x, y \in V(T)$, we denote the least common ancestor of $x$ and $y$ by $\lca(x, y)$, and note that $x$ is an ancestor of $y$ if and only if $\lca(x, y) = x$. Suppose $x$ is a strict ancestor of $y$, we define the $y$-directed child of $x$ as the child $z$ of $x$ such that $z$ is on the $y$-$x$ path. We denote such $z$ as $\dir(x, y)$.

\paragraph{Tree decompositions.}  A \textit{tree decomposition} of a graph $G$ is a pair $\mathcal{T} = (T, \bag)$ where $T$ is a rooted tree with root $r \in V(T)$ and $\bag: V(T) \rightarrow 2^{V(G)}$ is a mapping that assigns every tree node $x \in V(T)$ a subset $\bag(x) \subseteq V(G)$ called a bag. Furthermore, the following two conditions must be satisfied.

\begin{enumerate}
    \item For every vertex $v \in G$, the set of nodes $x \in T$ where $v \in \bag(x)$ induces a subtree of $T$.
    \item For every edge $(u, v) \in E(G)$, there exists a node $x \in T$ such that $u, v \in \bag(x)$.
\end{enumerate}

For a node $x \in T$, we additionally define the following notions.

\begin{itemize}
    \item The \textit{adhesion} of $x$ is $\adh(x) = \bag(\parent(x)) \cap \bag(x)$.
    \item The \textit{margin} of $x$ is $\mrg(x) = \bag(x) \setminus \adh(x)$.
    \item The \textit{cone} at $x$ is the union of all $\bag(y)$ where $y$ is a descendant of $x$, denoted by $\cone(x)$.
    \item The \textit{component} at $x$ is $\comp(x) = \cone(x) \setminus \adh(x)$.
\end{itemize}

The \textit{adhesion} of a tree decomposition $\mathcal{T}$ is defined as the size of the largest adhesion, which is $\max_{x \in V(T)} |\adh(x)|$.

Without loss of generality, we will assume that any tree decomposition we work with is \textit{regular}, which satisfies the following conditions for every non-root node $x \in V(T)$.

\begin{enumerate}
    \item The margin of $x$ is non-empty.
    \item The induced graph $G[\comp(x)]$ is connected.
    \item Every vertex $u \in \adh(x)$ has a neighbor $v \in \comp(x)$.
\end{enumerate}

The last two conditions would imply that for every pair of vertices $u, v \in \adh(x)$ and for any node $x$, there must exist a $u$-$v$ path such that every vertex along this path except $u, v$ are contained in $\comp(x)$.

\paragraph{Unbreakable decompositions.} A \textit{separation} in a graph $G$ is two subsets of vertices $L, R \subseteq V(G)$ where $L \cup R = V(G)$ such that after removing its \textit{separator} $L \cap R$ from $G$, there will be no edge between the vertices in $L \setminus R$ and $R \setminus L$ in $G \setminus(L \cap R)$. The \textit{order} of a separation is the size of its separator. For a subset $X \subseteq V(G)$, we say that $X$ is $(q, k)$-\textit{unbreakable} in $G$ if for every separation $(L, R)$ of order at most $k$, we either have $|L \cap X| \leq q$ or $|R \cap X| \leq q$.

A $(q, k)$-\textit{unbreakable decomposition} of a graph $G$ is a tree decomposition $\mathcal{T} = (T, \bag)$ where for every node $x \in V(T)$, $\bag(x)$ is $(q, k)$-unbreakable in $G$. Additionally, we say that $\mathcal{T}$ is \textit{strongly} $(q, k)$-unbreakable if every $\bag(x)$ is $(q, k)$-unbreakable in $G[\cone(x)]$. We will use the following theorem to compute an unbreakable decomposition in this paper.

\begin{ptheorem}[\cite{korhonen2024}]\label{thm:unbreakable-decomp}
    Given a graph $G$, there is a deterministic algorithm that computes a strongly $(k, k)$-unbreakable decomposition with adhesion $k$ and height $O(\log n)$ in $k^{O(k^2)} n + O(m)$ time.
\end{ptheorem}

\paragraph{Bag graphs.} Let $\mathcal{T} = (T, \bag)$ be a regular tree decomposition of a graph $G$. For any node $x \in V(T)$, we define its \textit{bag graph} in the following way, denoted by $\bgraph(x)$. We start by initializing $\bgraph(x)$ to be the subgraph $G[\bag(x)]$, and we call the edges in this graph \textit{normal}. Next, for each $z \in \children(x)$, we add a copy of the edge $(u, v)$ for every pair of vertices $u, v \in \adh(z)$ in $\bgraph(x)$. We call these edges \textit{adhesion edges supported by $z$}, and they would be labeled as such. Note that we allow multiple edges between some pair of $u, v \in \bag(x)$ supported by different children of $x$, which means that $\bgraph(x)$ is in fact a multigraph.

We claim that the total size of all bag graphs is linear to the total size of $G$.

\begin{plemma}\label{lemma:total-bgraph-size}
    Let $\mathcal{T} = (T, \bag)$ be a regular tree decomposition of $G$ of adhesion $q$. Then we have the following.
    \begin{enumerate}
        \item $\sum_{x \in V(T)} |V(\bgraph(x))| \leq O(qn)$
        \item $\sum_{x \in V(T)} |E(\bgraph(x))| \leq O(m + q^2n)$
    \end{enumerate}
\end{plemma}

\proof For the first bound, note that the vertex set of each $\bgraph(x)$ is $\bag(x)$, and by definition we have $|\bag(x)| = |\mrg(x)| + |\adh(x)|$ for each bag. The margins of $\mathcal{T}$ are all disjoint, so the total size of all margins would be exactly $n$. Note that since $\mathcal{T}$ is regular, each margin is non-empty. Since each vertex is in the margin of exactly one node, we have that $|T| \leq n$ as otherwise, there would exist at least one node $x$ such that $\bag(x) = \adh(x)$. We also have $|\adh(x)| \leq q$ for all nodes $x$, then

\[\sum_{x \in V(T)} |V(\bgraph(x))| = \sum_{x \in V(T)} |V(\mrg(x))| + \sum_{x \in V(T)} |V(\adh(x))| \leq n + n \cdot q = (q + 1)n = O(qn)\]

For the second bound, in each $\bgraph(x)$, first consider all edges whose two endpoints are both in $\mrg(x)$, we call these edges \textit{vital}. Note that each edge in $G$ can be vital in at most one bag graph, so the total number of vital edges in all bag graphs is at most $m$. Each vertex in $\bgraph(x)$ has at most $q$ edges incident to $\adh(x)$ as $|\adh(x)| \leq q$. For each child $z$ of $x$, the number of adhesion edges supported by $z$ in $\bgraph(x)$ is at most $q^2$ as $|\adh(z)| \leq q$. Note that each node has at most one parent, thus

\[\sum_{x \in V(T)} |E(\bgraph(x))| = m + \sum_{x \in V(T)} (q \cdot |V(\bgraph(x))| + q^2 \cdot |\children(x)|)\]
\[= m + q \cdot \sum_{x \in V(T)}|V(\bgraph(x))| + q^2 \cdot \sum_{x \in V(T)} |\children(x)|\]
\[\leq m + q(q + 1)n + q^2 |T| \leq m + q(q + 1)n + q^2 n = O(m + q^2 n)\]
\\
This concludes the proof.

\paragraph{Restricted bag graphs.} For a regular tree decomposition $\mathcal{T} = (T, \bag)$ of $G$, a tree node $x$, and a subset $S \subseteq V(G)$, we define the \textit{$S$-restricted bag graph} of $x$, which is denoted by $\bgraph_S(x)$ and obtained in the following way. We start with $\bgraph(x)$, and for each child $z$ of $x$, we consider all adhesion edges supported by $z$. We remove a $z$-supported adhesion edge $(u, v)$ from $\bgraph(x)$ if and only if there does not exist a $u$-$v$ path in $G[\cone(z)] \setminus S$. We obtain $\bgraph_S(x)$ after considering all adhesion edges. Note that $\bgraph_{\emptyset}(x) = \bgraph(x)$ as $\mathcal{T}$ is regular.

We now consider some vital properties of restricted bag graphs in unbreakable decompositions, captured in the following lemma.

\begin{plemma}\label{lemma:bgraph-structure}
    Let $\mathcal{T} = (T, \bag)$ be a strongly $(q, k)$-unbreakable decomposition of $G$ that is regular. Let $x \in V(T)$ and $S \subseteq V(G)$ where $|S| \leq k$, then the following properties must hold.

    \begin{enumerate}
        \item For any $u, v \in \bag(x)$, they are connected in $G[\cone(x)] \setminus S$ if and only if they are connected in $\bgraph_S(x)$.
        \item Suppose $C \subseteq \bag(x)$ is a connected component in $\bgraph_S(x)$ where $|C| > q$. Let $C_1,...,C_{\ell}$ be the connected components in $\bgraph_S(x) \setminus C$, then we have $\sum_{i = 1}^{\ell} |C_i| \leq q$.
    \end{enumerate}
\end{plemma}

\proof Note that the property 1 directly follows from the construction of $\bgraph_S(x)$. For property 2, let $C' = \bigcup_{i = 1}^{\ell} C_i$, and suppose that $|C'| > q$. By property 1, $C$ and $C'$ must be disconnected in $G[\cone(x)] \setminus S$. Let $L$ be the connected component in $G[\cone(x)] \setminus S$ such that $C \subseteq L$, and let $R = (\cone(x) \setminus S) \setminus L$. Note that $C' \subseteq R$. Then, $(L \cup S, R \cup S)$ is a separation in $G[\cone(x)]$ of order at most $k$ with separator $S$, where $|(L \cup S) \cap \bag(x)| = |C| > q$ and $|(R \cup S) \cap \bag(x)| = |C'| > q$, contradicting the unbreakability condition. Thus, we must have $|C'| \leq q$. This finishes the proof.

\paragraph{Vertex-failure connectivity oracle.} We introduce an existing vertex-failure connectivity oracle presented in \cite{kosinas2023} that we will utilize in the preprocessing phase, summarized in the following theorem.

\begin{ptheorem}[\cite{kosinas2023}]\label{thm:edge-failure-connectivity-oracle}
    There is a deterministic vertex-failure connectivity oracle that uses $O(k m \log n)$ space, $O(k m \log n)$ preprocessing time, $O(k^4 \log n)$ update time, and $O(k)$ query time.
\end{ptheorem}

\paragraph{Shortcutting of trees.} We now consider adding shortcuts to a tree. Formally, given a tree $T$, we want to add a small number of additional edges to $T$ such that for any nodes $x, y \in V(T)$ where $x$ is an ancestor of $y$, one can reach $x$ from $y$ via $h$ non-reversing edges where $h$ is a constant. The non-reversing edges mean that every edge taken from $y$ to reach $x$ travels towards the root. This is called a \textit{shortcutting} of the tree $T$ with \textit{hop-diameter} $h$. From \cite{bodlaender1994,le2025} we know that a good shortcutting can be computed efficiently on paths and trees. Explicitly, we can utilize the following theorem.

\begin{ptheorem}[\cite{le2025, bodlaender1994}]\label{thm:tree-shortcutting}
    Every tree of $n$ nodes admits a shortcutting with hop-diameter $h$ that uses $O(n  \alpha_{h/2 + 1} (n))$ edges, which can be constructed in time $O(n  \alpha_{h/2 + 1} (n))$.
\end{ptheorem}

Although implicit in \cite{le2025,bodlaender1994}, given any $x,y$, one can compute an $h$-hop path with shortcut edges in $O(h)$ constant time. Note that here $\alpha_c(n)$ means inverse Ackermann function as in \cite{le2025}. For example, $\alpha_{2} (n) = O(\log n)$ and $\alpha_4(n) = O(\log^*n)$, and so on.

%\benyu{Discuss construction time, and how to find explicit hop-$c$ paths via RMQ if time permits.}

\section{Preparing the connectivity oracle}\label{sec:prep}
In this section we prove the space and preprocessing time of \cref{thm:main}. We start by constructing a strongly $(k, k)$-unbreakable decomposition $\mathcal{T} = (T, \bag)$ of $G$ with adhesion at most $k$ and is regular using \cref{thm:unbreakable-decomp}. The construction takes time $k^{O(k^2)} n + O(m)$.

Next, we need a structure that, given any two tree nodes $x, y \in V(T)$, answers $\lca(x, y)$ and $\dir(x, y)$ queries in constant time. We use the following lemma given in \cite{pilipczuk2021algorithms}, which relies on the data structure for $\lca$ queries presented in \cite{HarelT84}.

\begin{plemma}\label{lemma:lca-dir}
    Given a tree $T$, one can compute in time $O(|V(T)|)$ a data structure that may answer the $\lca(x, y)$ and $\dir(x, y)$ queries in $O(1)$ time.
\end{plemma}

\subsection{Torsos and profiles}

We now design data structures that allow quick evaluations of the connectivity between vertices in the adhesions of tree nodes that speed up update and query times.

First, we introduce the notion of torsos. For any $x, y \in V(T)$ where $x$ is a strict ancestor of $y$, we define the \textit{torso graph between $x$ and $y$}, denoted by $\torso(x, y)$, in the following way. The vertex set of $\torso(x, y)$ is $\adh(x) \cup \adh(y)$. Any two vertices $u, v \in V(\torso(x, y))$ are adjacent if and only if there exists a $u$-$v$ path in $G[\cone(x) \setminus \comp(y)]$ whose interior is disjoint from $V(\torso(x, y))$. Note that any torso graph consists of at most $2k$ vertices.

We often need torsos in the update and query phases. To do this quickly, we need a data structure that answers torso queries in constant time. We first show that we can compute all torsos quickly.

\begin{pclaim}\label{claim:computing-all-torsos}
    One can compute in $O(k^3 n \log^2 n + k^6 n \log n)$ time $\torso(x, y)$ for all $x, y \in V(T)$ where $x$ is a strict ancestor of $y$.
\end{pclaim}

\proof We first build $\torso(x, z)$ for every pair $x, z \in V(T)$ where $x$ is the parent of $z$. Note that $V(\torso(x, z)) = \adh(x) \cup \adh(z) \subseteq \bag(x)$. Consider the $\comp(z)$-restricted bag graph of $x$, recall from \cref{lemma:bgraph-structure} that two vertices in $\bag(x)$ are connected in $G[\cone(x) \setminus \comp(z)]$ if and only if they are connected in $\bgraph_{\comp(z)}(x)$. Any pair of vertices in $\torso(x, z)$ are adjacent if and only if there is a path in $\cone(x) \setminus \comp(z)$ that is internally vertex disjoint from $\adh(x) \cup \adh(z)$. Thus to determine the adjacency of vertices in $V(\torso(x, z))$ in $G[\cone(x) \setminus \comp(z)]$, we only need to determine if such path exists in $\bgraph_{\comp(z)}(x)$ for each pair of vertices. We can obtain $\bgraph_{\comp(z)}(x)$ explicitly by removing all $z$-supported adhesion edges in $\bgraph_{\emptyset}(x)$ since two vertices $u, v \in \adh(z)$ are connected in $G[\cone(z) \setminus \comp(z)]$ if and only if they are adjacent in $G[\adh(z)]$, and these adjacencies are represented by the normal edges in $\bgraph_{\comp(z)}(x)$.

These torsos are computed in the following way. For each $x \in V(T)$, consider a variant of $\bgraph_{\emptyset}(x)$ where for every adhesion edge $(u, v)$, we delete it and replace it with the edges $(u, w), (w, v)$ where $w$ is a new vertex. There are at most $k^2 \cdot |\children(x)|$ adhesion edges in $\bgraph_{\emptyset}(x)$ since each child $z$ of $x$ has adhesion of size at most $k$ thus supports at most $k^2$ adhesion edges. Thus the number of vertices in this variant is at most $|V(\bgraph_{\emptyset}(x))| + k^2 \cdot |\children(x)|$ and it has at most $|E(\bgraph_{\emptyset}(x))| + k^2 \cdot |\children(x)|$ edges. We then use \cref{thm:edge-failure-connectivity-oracle} to construct an vertex-failure connectivity oracle on this variant of $\bgraph_{\emptyset}(x)$ for each $x \in V(T)$ and set $d = 2k= O(k)$.

To build $\torso(x, z)$ we first consider each pair of vertices $u, v \in \adh(z)$. Suppose $w$ is the added vertex in the variant where $(u, w), (w, v)$ replaced the $z$-supported adhesion edge $(u, v)$. To decide if any such pair $u, v$ are adjacent in $\torso(x, z)$, we update the oracle of $x$ by giving the set $D = (\adh(x) \cup \adh(z) \cup \{w\}) \setminus \{u, v\}$ of failed vertices, and query if $u, v$ are still connected after the update. We add the edge $(u, v)$ to $\torso(x, z)$ if they are still connected. Note that this is correct because the query outputs connected if and only if there is a $u$-$v$ path in $\bgraph_{\comp(z)}(x)$ that is internally disjoint from $\adh(x) \cup \adh(z)$ as we deleted all vertices in $\adh(x) \cup \adh(z)$ and all $z$-supported adhesion edges by also deleting $w$. We do a similar process for the cases of $u \in \adh(z), v \in \adh(x)$ and $u, v \in \adh(x)$ with the only difference being setting $D = (\adh(x) \cup \adh(z)) \setminus \{u, v\}$ instead as there is no such $w$ in these cases. Then we have correctly computed $\torso(x, z)$. Additionally, the updates are valid since $D \leq 2k$ for all updates.

We now analyze the runtime of this procedure. Initializing the oracle for every variant of $\bgraph_{\emptyset}(x)$ by \cref{thm:edge-failure-connectivity-oracle} in total takes time:
\[\sum_{x \in V(T)} 2k \cdot (|E(\bgraph_{\emptyset}(x))| + k^2 \cdot |\children(x)|) \cdot \log (|V(\bgraph_{\emptyset}(x))| + k^2 \cdot |\children(x)|)\]
\[\leq 2k \cdot \sum_{x \in V(T)} (|E(\bgraph_{\emptyset}(x))| + k^2 \cdot |\children(x)|) \cdot O(\log n)\]
\[= O(k \log n) \cdot (\sum_{x \in V(T)} |E(\bgraph_{\emptyset}(x))| + k^2 \cdot \sum_{x \in V(T)} |\children(x)|)\]
\[\leq O(k \log n) \cdot ((m + k^2 n) + k^2 n) = O(k^3 n \log n)\]
\\
Each update step takes at most $O(k^4 \log n)$ time and each query step takes $O(k)$ time. There are $n$ such $\torso(x, z)$, and to build each we perform an update followed by a query for $k^2$ times, so in total these steps take time $O(k^6 n \log n)$. Thus we can construct all such $\torso(x, z)$ in $O(k^6 n \log n)$ time.

Now we build all $\torso(x_{\ell}, x_1)$ where $x_{\ell}$ is a strict ancestor of $x_1$ but $x_{\ell} \neq \parent(x_1)$. Let $x_1,x_2,..., x_{\ell}$ be the $x_1$-$x_{\ell}$ path in $T$. We construct the graph $H$ on the vertex set $\bigcup_{i = 1}^{\ell - 1} V(\torso(x_{i + 1}, x_i))$ and the edge set $\bigcup_{i = 1}^{\ell - 1} E(\torso(x_{i + 1}, x_i))$. Note that any $u, v \in V(\torso(x_{\ell}, x_1))$ have a path in $G[\cone(x_{\ell}) \setminus \comp(x_1)]$ that is interior disjoint from $\adh(x_1) \cup \adh(x_{\ell})$ if and only if such a path exists in $H$ since each $\torso(x_{i + 1}, x_i)$ maps the connectivity within $G[\cone(x_{i + 1}) \setminus \comp(x_i)]$ and any path that goes through $\torso(x_i, x_{i - 1})$ and $\torso(x_{i + 1}, x_i)$ must pass some vertex in $\adh(x_i)$. This way we can construct $\torso(x_{\ell}, x_1)$ by performing BFS in $H$ starting from each $v \in V(\torso(x_{\ell}, x_1))$, avoiding other vertices in $\adh(x_{\ell}) \cup \adh(x_1)$, to determine which other vertices $u \in V(\torso(x_{\ell}, x_1))$ are reachable via a path that is internally disjoint from $\adh(x_{\ell}) \cup \adh(x_1)$, and add the edge $(u, v)$ for all such reachable $u$.

Since each torso has at most $k^2$ edges and $\ell \leq O(\log n)$ by the depth of $T$, each BFS step takes $O(|E(H)|) \leq O(k^2 \log n)$ time. There are at most $n \log n$ pairs of $x_{\ell}, x_1$ and to construct each $\torso(x_{\ell}, x_1)$ we need to perform $O(k)$ BFS steps, so we can construct all these torsos in time $n \log n \cdot O(k) \cdot O(k^2 \log n) = O(k^3 n \log^2n)$. Combining with the previous procedure, we can build all torso graphs in $O(k^3 n \log^2 n + k^6 n \log n)$ time, completing the proof.\\

Now we design the following structure that outputs torsos in constant time when being queried.

\begin{plemma}\label{lemma:torso-query}
    There is a data structure that uses $2^{O(k^2)} + O(k^2n \alpha_c(n))$ space, initializes in $O(k^3 n \log^2 n + k^6 n \log n)$ time, and answers the query in $O(1)$ time: given $x, y \in V(T)$ where $x$ is a strict ancestor of $y$, outputs a pointer to the graph $\torso(x, y)$.
\end{plemma}

%\benyu{I still believe the total space here is $O(nk\alpha_{c/2+1}(n))$. Does each pointer here require $\Omega(k)$ space?}

\proof Each $\torso(x, y)$ has at most $2k = O(k)$ vertices so storing it takes $O(k^2)$ space. Let $S_t$ be the set of all possible graphs on $2k$ vertices, then $|S_t| = 2^{O(k^2)}$, and we can store $S_t$ using $O(k^2) \cdot 2^{O(k^2)}$ space. We use \cref{claim:computing-all-torsos} to compute all possible torsos. For each computed $\torso(x, y)$ we keep a pointer to $G_t \in S_t$ where $G_t = \torso(x, y)$ for now. 

Consider constructing a shortcutting of $T$ using $O(n \alpha_c(n))$ edges with hop-diameter $2(c + 1)$ by calling \cref{thm:tree-shortcutting}. We aim to only store a pointer to $\torso(x, y)$ for any $x, y \in V(T)$ where the edge $(x, y)$ is in the shortcutting. This implies that we only need to store $O(n \alpha_{c}(n))$ pointers of size $O(k^2)$. To achieve constant time query, we build and store a two dimensional look-up table $U$ indexed by pairs of torsos. For each computed $\torso(x, y) = G_{t, 1}$ and $\torso(y, z) = G_{t, 2}$, we store the pointer to $G_t$ in the entry $U(G_{t, 1}, G_{t, 2})$ where $G_t = \torso(x, z)$. Then this table has at most $(2^{O(k^2)})^2 = 2^{O(k^2)}$ entries. Given any torso query $x, y \in V(T)$ where $x$ is a strict ancestor of $y$, suppose $y, x_1, x_2,...,x_h, x$ is the $y$-$x$ path in the shortcutting of $T$, so $h \leq 2(c + 1)$. We can obtain $\torso(x_2, y)$ by looking up $U(\torso(x_2, x_1), \torso(x_1, y))$, and then we can obtain $\torso(x_3, y)$ by looking up $U(\torso(x_3, x_2), \torso(x_2, y))$. Thus we can get $\torso(x, y)$ by repeating this step $h$ times along the path. Each look-up step takes constant time, so obtaining $\torso(x, y)$ takes $O(h) = O(1)$ time.

Therefore, the structure uses $2^{O(k^2)} + O(k^2n  \alpha_c(n))$ space to store all possible torsos, the table $U$, and the pointers. The initialization time is dominated by computing all torsos, which is $O(k^3 n \log^2 n + k^6 n \log n)$, and torso queries can be answered in $O(1)$ time, finishing the proof.\\

Next we introduce profiles, which maintain the connectivity of an adhesion within its subtree. For a node $x \in V(T)$, we define the \textit{profile} of $x$, denoted by $\profile(x)$, which is a graph on the vertex set $\adh(x)$. Two vertices $u, v \in V(\profile(x))$ are adjacent if and only if they are connected in $G[\cone(x)] \setminus S$ where $S$ is the set of failed vertices. Note that by this definition every connected component in each profile must be a clique as connectivity is transitive. As profiles are dependent on the set $S$ of failed vertices, we do not build profiles directly in this section, but they will be constructed and used extensively in the update and query phases. In this section, we design the following data structure to compute profiles quickly from existing profiles and torsos.

\begin{plemma}\label{lemma:combining-profiles-with-torsos}
    There is a data structure that uses $O(k^2) \cdot 2^{O(k^2)}$ space, initializes in time $O(k^2) \cdot 2^{O(k^2)}$, and answers the query in $O(1)$ time: given $\profile(y)$ and $\torso(x, y)$, where $x$ is a strict ancestor of $y$ and $\comp(x) \setminus \cone(y)$ is disjoint from $S$, output the following.
    \begin{enumerate}
        \item The graph $\profile(x)$.
        \item A coloring $\col: \adh(x) \cup \adh(y) \rightarrow [O(k)]$ where $\col(u) = \col(v)$ if and only if $u, v \in \adh(x) \cup \adh(y)$ are connected in $G[\cone(x)] \setminus S$.
        \item A inverse coloring $\invcol: [O(k)] \rightarrow \adh(x)$ such that $\invcol(i) = v$ where $v \in \adh(x)$ is an arbitrary vertex with $\col(v) = i$.
    \end{enumerate}
\end{plemma}

\proof First note that $|V(\profile(x))| \leq k$ for any $x \in V(T)$, so $|E(\profile(x))| \leq O(k^2)$. There are $2^{k^2}$ possible graphs on $k$ vertices, we compute and store them explicitly in $k^2 \cdot 2^{k^2}$ time and space. Recall that each torso has at most $2k$ vertices, so storing all possible torsos take $O(k^2) \cdot 2^{O(k^2)}$ time and space.

Let $S_t, S_p$ be the set of all torsos and profiles. Next, we construct a 2-dimensional look-up table $U$, indexed by torsos and profiles. For each $G_t \in S_t$ and $G_p \in S_p$ we compute entry be $U(G_t, G_p)$ in the following way. Construct graph $H$ on the vertex set $V(G_t) \cup V(G_p)$ and the edge set $E(G_t) \cup E(G_p)$. Then, perform BFS over $H$ to identify its connected components. We construct the graph $H'$ on vertex set $V(H') =  V(G_t) \setminus V(G_p)$, and each pair of vertices $u, v$ in this vertex set are adjacent if and only if they are in the same connected component in $H$. We assign an arbitrary coloring $\col: V(G_t) \rightarrow [O(|V(G_t)|)]$ such that only vertices in the same connected component in $H$ have the same color. For each color $i$ used, we assign $\invcol(i) = v$ where $v$ is an arbitrary vertex $v \in V(H')$ such that $\col(v) = i$ (or $\invcol(i) = \emptyset$ if such $v$ does not exist). Finally we set the entry $U(G_t, G_p) = (H', \col, \invcol)$.

It is clear that computing one entry takes $O(k^2)$ time. Suppose that on a valid query $\torso(x, y)$ and $\profile(y)$ we have the entry $U(\torso(x, y), \profile(y)) = (H',\col, \invcol)$. By construction, $H'$ has vertex set exactly $\adh(x) = V(\profile(x))$, and two vertices $u, v \in \adh(x)$ are adjacent if and only if they are connected in $G[\cone(x)] \setminus S$, thus $H' = \profile(x)$. The coloring $\col$ and inverse coloring $\invcol$ are also correct by the same argument.

$U$ has $|S_t| \cdot |S_p| = 2^{O(k^2)} \cdot 2^{k^2} = 2^{O(k^2)}$ entries, each storing a graph of at most $k$ vertices that takes $O(k^2)$ time to compute along with a coloring and inverse coloring linear in its size. Thus, constructing the table $U$ takes $O(k^2) \cdot 2^{O(k^2)} = O(k^2) \cdot 2^{O(k^2)}$ time and space. On query  $\torso(x, y)$ and $\profile(y)$, the structure outputs $U(\torso(x, y), \profile(y)) = (\profile(x), \col, \invcol)$ by looking up the table in constant time. This completes the proof.

\subsection{Patch sets, processed bag graphs, and neighbor sets}

To maintain and determine connectivity in update and query, the data structure relies heavily on bag graphs. More specifically, for every node $x \in V(T)$, we store a version of $\bgraph_{\emptyset}(x)$ that is preprocessed, along with the \textit{patch set} of $x$, denoted by $\pset(x)$. Each $\pset(x)$ contain \textit{patches}, where a \textit{patch} $P \subseteq \bag(x)$ is simply a subset of vertices. We add patches to $\pset(x)$ in the following way. Start by choosing an arbitrary vertex $v \in \bag(x)$ and perform BFS in $\bgraph_{\emptyset}(x)$ starting from $v$ until $k + 1$ vertices, including $v$, are visited. Let $P_1$ be a set that contains exactly these $k + 1$ visited vertices, remove $P_1$ from $\bgraph_{\emptyset}(x)$, include $P_1$ in $\pset(x)$, and repeat this process in $\bgraph_{\emptyset}(x) \setminus P_1$. This is repeated until in $\bgraph_{\emptyset}(x) \setminus (\bigcup_{P_i \in \pset(x)} P_i)$ there is no connected component of size at least $k + 1$. Finally, for any remaining connected component $C$, we include $C$ as one patch in $\pset(x)$.

We say a patch $P \in \pset(x)$ is \textit{big} if $|P| = k + 1$, and we say that it is \textit{small} otherwise. Moreover, we call two patches $P_i, P_j \in \pset(x)$ \textit{adjacent} to each other if and only if there exist some $u \in P_i$ and $v \in P_j$ such that $u$ and $v$ are adjacent in $\bgraph_{\emptyset}(x)$. From the construction process, the following structural claim is clear.

\begin{pclaim}\label{claim:pset-structure}
    For any $x \in V(T)$ and its computed $\pset(x)$, the following properties must hold.
    \begin{enumerate}
        \item All patches in $\pset(x)$ are disjoint, and $\bigcup_{P_i \in \pset(x)} P_i = \bag(x)$.
        \item Any small patch $P \in \pset(x)$ can only be adjacent to big patches.
    \end{enumerate}
\end{pclaim}

We then proceed with a simple processing of bag graphs. For each vertex $v \in \bag(x)$, store in $\bgraph_{\emptyset}(x)$ a pointer along with $v$ that points to $P_v$, where $P_v \in \pset(x)$ is the patch such that $v \in P_v$.

Along side this, we also store some neighbor sets of the patches. Let $N(P)$ be the set of at most $k + 1$ vertices $v \in \bag(x)$ where $v \notin P$ and there exists a normal edge $(u, v)$ in $\bgraph_{\emptyset}(x)$ such that $u \in P$ (if there are more than $k + 1$ such $v$, choosing an arbitrary subset of $k + 1$ of them suffices). We call $N(P)$ the set of \textit{normal neighbors} of $P$. 

Moreover, for any patch $P \in \pset(x)$, let $N_{\adh}(P)$ be the set of neighbors $v$ such that there is an adhesion edge $(u, v)$ in $\bgraph_\emptyset(x)$ where $u \in P$ and $v \notin P$. We call this the set of \textit{adhesion neighbors} of $P$. Each set $N_{\adh}(P)$ stores at most $k^3 + 1$ such vertices, and we implement it as a doubly-linked list such that deletion can be done in constant time. We store along with each vertex $v \in N_{\adh}(P)$ a counter indicating the number of adhesion edges $(u, v)$ where $u \in P$.

Now, suppose $P_1,...,P_{\ell} \in \pset(x)$ are all the small patches. We store with each $P_i \in \{P_1,...,P_{\ell} \}$ the set $N(P_i)$ and $N_{\adh}(P_i)$. Furthermore, we claim that we can compute and store all patch sets, processed bag graphs, and neighbor sets of small patches in time and space linear to the size of $G$.

\begin{plemma}\label{lemma:pset-bgraph-time-and-space}
    One can in time $O(k^3n)$ and space $O(k^2 n)$ compute and store the following.
    \begin{enumerate}
        \item For every $x \in V(T)$, the patch set $\pset(x)$ and the processed $\bgraph_{\emptyset}(x)$.
        \item For every $x$ and each small patch $P \in \pset(x)$, the neighbor sets $N(P)$ and $N_{\adh}(P)$.
    \end{enumerate}
\end{plemma}

\proof Recall from \cref{lemma:total-bgraph-size} that in $\mathcal{T}$ the total size of all bags is at most $O(qn) = O(kn)$ and the total size of all bag graphs is at most $O(m + k^2n)$. From \cref{claim:pset-structure} we know that for a $\pset(x)$ all patches are disjoint and cover $\bag(x)$, thus $|\bigcup_{P \in \pset(x)} P| = |\bag(x)|$, so the total space to store all patch sets would be the same as the total size of all bags, which is $O(kn)$. We compute $\pset(x)$ by performing BFS procedures in $\bgraph_{\emptyset}(x)$ such that each edge is visited only once, so the total time required to computed all patch sets would be $\sum_{x \in V(T)} |E(\bgraph_{\emptyset}(x))| = O(m + k^2 n)$. Moreover, when processing each $\bgraph_{\emptyset}(x)$, we only add one pointer for each vertex, so the space to store all processed bag graphs would still be $O(m + k^2n) = O(k^2n)$. This process can be done while we are computing patch sets.

In each $\pset(x)$ there can be at most $|\bag(x)|$ small patches. For each small patch $P$, at most $k + 1$ vertices can be in $N(P)$. Hence storing all normal neighbor sets across all patch sets would require at most $O(k \cdot kn) = O(k^2n)$ space. To compute $N(P)$, we visit $2k + 1$ neighbors of each $v \in P$ if able, and for each neighbor $u$ of $v$, we include it in $N(P)$ if it is not already included and $u \notin P$. We stop early when $|N(P)| = k + 1$. This guarantees that if $P$ has at most $k$ neighbors, all of them would be in $N(P)$, and otherwise $N(P)$ would contain $k + 1$ arbitrary neighbors of $P$. This takes $O(k^2)$ time for each small patch $P$, so in total across all patch sets this takes $O(k^2 \cdot kn) = O(k^3n)$ time.

For each $x$, we next check each adhesion edge $(u, v)$ in $\bgraph_{\emptyset}(x)$. Let $P_u, P_v$ be the patches that contain $u, v$ respectively. If $P_u = P_v$, we discard this edge and continue to the next one. Otherwise, we add $v$ to $N_{\adh}(P_u)$ and $u$ to $N_{\adh}(P_v)$, incrementing the counters accordingly. We do not add an edge to a set if it already contains $k^3 + 1$ vertices. There are at most $k^2 \cdot |\children(x)|$ adhesion edges in $\bgraph_{\emptyset}(x)$ since each child of $x$ supports at most $k^2$ adhesion edges. Thus the total number of adhesion edges for in bag graphs is $\sum_{x \in V(T)} k^2 \cdot |\children(x)| = k^2n = O(k^2n)$, and we can compute all outgoing adhesion edge sets with $O(k^2n)$ time and space.

Combining all of these, the total time required is $O(k^3n)$ while the total space usage is $O(k^2n)$, which proves the lemma.

\subsection{Auxiliary structure for single-child failure}\label{section:single-child-failure}

In this section we consider the special case of single-child failure. Suppose $S$ is the set of failed vertices, for a node $x \in V(T)$, we say that it is in the case of \textit{single-child failure} if $S$ is disjoint from $\bag(x)$ and there is exactly one $z \in \children(x)$ such that $\comp(z) \cap S \neq \emptyset$. We call such $z$ the \textit{single affected child} of $x$. To speed up query time, we need $O(k)$ time access to some vital information about $\bgraph_S(x)$ for nodes $x$ that are in this case, formally captured in the following lemma.

\begin{plemma}\label{lemma:single-child-failure-query}
    There is a data structure that uses $2^{O(k^2)} n$ space, initializes in $2^{O(k^2)}n$ time, and answers the query in $O(k)$ time: given a node $x \in V(T)$ in the case of single-child failure, $\profile(z)$ with $z$ being its single affected child, and a vertex $u \in \bag(x)$, output the following.
    \begin{enumerate}
        \item Let $C$ be the unique connected component in $\bgraph_S(x)$ where $|C| > k$, the set $L = \adh(x) \cap C$.
        \item The set of vertices $I$ where each $v \in I$ is not in $C$ and is connected to some $v' \in \adh(x)$, with a coloring $\col: I \rightarrow [O(k)]$ such that $\col(u) = \col(v)$ for any $u, v \in I$ if and only if they are connected in $\bgraph_S(x)$.
        \item A label indicating whether $u \in C$.
    \end{enumerate}
\end{plemma}

We defer the proof of this lemma to \cref{sec:single-child-failure} as we will use some key steps developed in \cref{sec:update-query} for the update phase to help with the computation. At this moment we will assume the correctness of this lemma, and we can think of it intuitively as a look-up table that stores all queried information and takes $O(k)$ time to check the entries.

\subsection{Data structure}
With all necessary supporting structures in place, we can describe the data structure used to prove \cref{thm:main}, which consists of the following parts.

\begin{enumerate}
    \item The unbreakable decomposition $\mathcal{T} = (T, \bag)$ computed using \cref{thm:unbreakable-decomp}. Each node $x \in V(T)$ also stores its depth in $T$. Additionally, each vertex $v \in G$ contains a pointer that points to the unique tree node $x_v$ where $v \in \mrg(x_v)$.
    \item The data structure for $\lca$ and $\dir$ queries of \cref{lemma:lca-dir}.
    \item \begin{enumerate}
            \item[(a)] The torso query structure given by \cref{lemma:torso-query}.
            \item[(b)] The profile computing structure given by \cref{lemma:combining-profiles-with-torsos}.
        \end{enumerate} 
    \item For each node $x \in V(T)$, the patch set $\pset(x)$, processed bag graph $\bgraph_{\emptyset}(x)$, and additionally for each patch $P \in \pset(x)$, its neighbor set $N(P)$. These are computed using \cref{lemma:pset-bgraph-time-and-space}.
    \item The auxiliary data structure for the case of single-child failure given by \cref{lemma:single-child-failure-query}. 
\end{enumerate}

We now analyze the total space usage and processing time of this data structure. It is clear that all structures use space linear in $G$ except 3(a), which uses $2^{O(k^2)} +  O(k^2n \alpha_c(n))$ space. Also, structure 5 uses space that is linear in $n$ but exponential in $k$. Thus their total space usage would be $2^{O(k^2)}n + O(k^2n \alpha_c(n))$. The total preprocessing time is dominated by computing the unbreakable decomposition and initializing structure 3(a), which would be $k^{O(k^2)} n + O(k^3 n \log^2 n + k^6n \log n)$ in total\footnote{Recall we assume $m=O(kn)$ from \cite{nagamochi1992computing}. If $m>kn$, we actually take an extra $O(m)$ time in sparsification.}. Hence this proves the guarantees in \cref{thm:main}.

\section{Update and query algorithms}\label{sec:update-query}
In this section we prove the update and query time in \cref{thm:main} using the data structure defined and preprocessed in \cref{sec:prep}. Recall that in the update phase we are given a set $S \subseteq V(G)$ of failed vertices where $|S| \leq k$ and we update the data structure accordingly, and in the query phase we are given a pair of vertices $u, v \in V(G)$ and answer whether they are connected in $G \setminus S$. Note that there can be an update followed by multiple queries. Let us begin by describing the update algorithm.

\subsection{Update phase}

Suppose we are given $S$, we first label all $v \in S$ as "failed" in $G$ without deleting them or their incident edges. We then compute the set $X \subseteq V(T)$ of all nodes $x \in T$ that satisfy $\mrg(x) \cap S \neq \emptyset$. This can be computed in $O(k)$ time by including $x_v$ in $X$ for each $v \in S$ by checking its pointer, where $x_v$ is the unique node such that $v \in \mrg(x_v)$. Let $r$ be the root of $T$, if $r \notin X$, we also include $r$ in $X$. Note that $|X| \leq k + 1$. Additionally, let $Y$ be the lowest common ancestor closure of $X$, so $|Y| \leq 2|X| - 1 \leq 2k - 1 = O(k)$. We can compute $Y$ in $O(k^2)$ time by computing $\lca(x, y)$ using \cref{lemma:lca-dir} and including it in $Y$ for every pair of nodes $x, y \in X$.

We say that a node $x \in T(V)$ is \textit{important} if $x \in Y$. Furthermore, we call a node $x$ \textit{affected} if $\comp(x) \cap S \neq \emptyset$, otherwise we call it \textit{unaffected}. We begin by providing a algorithm that allows us to determine connectivity between a subset of vertices of $\bag(x)$ within $G[\cone(x)] \setminus S$ quickly for any important node $x$.

\begin{plemma}\label{lemma:subset-connectivity-in-bgraphs}
    For a node $x \in V(T)$, given $\bgraph_S(x)$ and a subset of vertices $V(H) \subseteq \bag(x)$, we can compute the graph $H$ in time $O(k^2 \cdot |V(H)|)$, where two vertices $u, v \in V(H)$ are adjacent in $H$ if and only if they are connected in $G[\cone(x)] \setminus S$.
\end{plemma}

\proof From \cref{lemma:bgraph-structure} we know that two vertices $u, v \in \bag(x)$ are connected in $G[\cone(x)] \setminus S$ if and only if they are connected in $\bgraph_S(x)$, and that there exists a unique connected component $C \subseteq \bag(x)$ in $\bgraph_S(x)$ with $|C| > k$. Then we build the edge set of $H$ as follows. We perform BFS in $\bgraph_S(x)$ starting from each $v \in V(H)$ until $k + 1$ vertices are reached, avoiding failed vertices. We add an edge $(u, v)$ in $H$ for every $u$ reached by $v$ where $u \in V(H)$. If more than $k$ vertices are reached, $v$ must be in $C$, and we mark it as such. After this step has been repeated for all vertices in $V(H)$, we add edges between all pairs of $u, v \in V(H)$ where both $u$ and $v$ have been marked that they are in $C$.

It is clear from the construction of $H$ that any two vertices are adjacent in $H$ if and only if they are connected in $\bgraph_S(x)$ and thus in $G[\cone(x)] \setminus S$. Each $(k + 1)$-step BFS takes $O(k^2)$ time, so processing all vertices in $V(H)$ takes $O(k^2 \cdot |V(H)|)$ time, completing the proof of this lemma.\\

Recall the definition of profiles, the next step is to compute $\profile(x)$ for all $x \in Y$.

\subsubsection{Computing important profiles}\label{sec:computing-important-profiles}

We will compute these profiles in a bottom-up manner. We first focus on a general node $x \in V(T)$. Suppose $Z \subseteq \children(x)$ is the set of all affected children of $x$. The following claim describes how $\profile(x)$ is computed given the profiles of all nodes in $Z$.

\begin{pclaim}\label{claim:building-profile-from-affected-children}
    Let $x \in V(T)$ be a node. Given the set $Z$ and $\profile(z)$ for all $z \in Z$, one can compute $\bgraph_S(x)$ and $\profile(x)$ in time $O(k^3)$.
\end{pclaim}

\proof First note that $|Z| \leq |S| \leq k$ as there are at most $k$ vertices in $S$ and the components of all affected children are pairwise disjoint. We will build $\bgraph_S(x)$ by removing edges from $\bgraph_{\emptyset}(x)$. By the definition of $\bgraph_S(x)$, only adhesion edges supported by the affected children need to be considered for deletion as $S$ is disjoint from the components of children in $\children(x) \setminus Z$. For every $z \in Z$, consider all adhesion edges supported by $z$. Recall that any such edge $(u, v)$ is not deleted in $\bgraph_S(x)$ if and only if $u$ and $v$ are connected in $G[\cone(z)] \setminus S$. Note that this is also exactly the definition of the edge set of $\profile(z)$, which is provided. Thus for each such edge $(u, v)$, we delete it if and only if $(u, v)$ is not in $\profile(z)$, which can be checked in constant time. Since each $|\adh(z)| \leq k$, there will be at most $|Z| \cdot k^2 \leq k \cdot k^2 = k^3$ edges to consider, thus we can construct $\bgraph_S(x)$ in $O(k^3)$ time.

To build $\profile(x)$, we need to decide for every pair $u, v \in \adh(x)$ whether they are connected in $G[\cone(x)] \setminus S$ or not. We use \cref{lemma:subset-connectivity-in-bgraphs} to compute $\profile(x)$ by giving as input the $\bgraph_S(x)$ that we just computed and $V(H) = V(\profile(x)) = \adh(x)$. Thus, we can obtain $H = \profile(x)$ in time $O(k^2 \cdot |\adh(x)|) \leq O(k^3)$. Note that building both $\bgraph_S(x)$ and $\profile(x)$ take $O(k^3)$ time in total, proving the claim.\\

Note that this lemma holds even when $x$ has no affected children since in this case $\bgraph_S(x) = \bgraph_{\emptyset}(x)$.

Next, we need to compute the profiles of all affected children $Z$ of $x$ in order to compute $\profile(x)$. The next lemma shows how $\profile(z)$ is computed for any non-important $z \in Z$ from the important profiles in its subtree.

\begin{pclaim}\label{claim:building-affected-child-profile}
    Let $z$ be an affected child of $x$ and suppose $z \notin Y$, given $\profile(y)$ for all $y \in Y$ that are contained in the subtree $T(z)$, one can compute $\profile(z)$ in time $O(k)$.
\end{pclaim}

\proof Let $y \in Y$ be the topmost important node in the subtree of $T$ rooted at $z$. It can be computed in $O(k)$ time by considering every $y \in Y$ and choosing the node  with the smallest depth such that $\lca(y, z) = z$. Note that $\comp(z) \setminus \cone(y)$ is disjoint from $S$ because otherwise there must be an important node $y' \in V(T(z)) \setminus V(T(y))$ and $\lca(y, y')$ will also be in $Y$ and have a smaller depth than $y$. Since we have $\profile(y)$ and can get $\torso(z, y)$ in constant time from \cref{lemma:torso-query}, we can use \cref{lemma:combining-profiles-with-torsos} to compute $\profile(z)$ in constant time in this situation. This completes the proof.\\

With these tools, we can proceed with computing all important profiles.

\begin{plemma}\label{lemma:computing-all-important-profiles}
    Given the set of important nodes $Y$, one can in time $O(k^4)$ compute $\profile(x)$ for all $x \in Y$.
\end{plemma}

\proof We begin by constructing the tree $T_Y$ where $V(T_Y) = Y$. For any two nodes $x, y \in Y$, $x$ is the parent of $y$ in $T_Y$ if and only if it is a strict ancestor of $y$ and the $x$-$y$ path in $T$ is disjoint from $Y$. Note that one can compute $T_Y$ in $O(k^2)$ time using $\lca$ queries on $Y$. Then, for each leaf node $x$ in $T_Y$, we can compute $\profile(x)$ using \cref{claim:building-profile-from-affected-children} by setting $Z = \emptyset$ as $x$ must have no affected children.

Next, for any non-leaf node $x \in Y$, we compute the set $Z$ of its affected children in $T$ by checking every $y \in Y$, and we include $\dir(x, y)$ in $Z$ if $\lca(x, y) = x$. Note that this takes $O(k^2)$ time for all $x$. With this, we can build all important profiles bottom-up by repeating the following process until all of them are computed. First, use \cref{claim:building-affected-child-profile} to compute the profiles of all affected children $z$ where the profile of every node in $V(T(z)) \cap Y$ is computed, then use \cref{claim:building-profile-from-affected-children} to compute any $\profile(x)$ whose affected children profiles have all been computed.

There are $O(k)$ important profiles to compute, so \cref{claim:building-profile-from-affected-children} is called $O(k)$ times, taking total time $O(k \cdot k^3) = O(k^4)$. Each affected child is computed from a unique important profile and the torso between them, so \cref{claim:building-affected-child-profile} is called $O(k)$ times, which take total time $O(k \cdot k) = O(k^2)$. Therefore, computing all important profiles takes $O(k^4) + O(k^2) = O(k^4)$ time, completing the proof.\\

Note that in the process of computing all important profiles, for each $x \in Y$, all profiles of its affected children are also computed.

\subsubsection{Computing adhesion connectivity graphs}\label{sec:computing-adhconn}

We now design a structure that keeps track of the connectivity between the profile of any important node and the profiles of its affected children within its restricted bag graph. Formally, for any $x \in Y$ with the set of its affected children $Z$, we define the \textit{adhesion connectivity graph}, denoted by $\adhconn(x)$, as the graph on the vertex set $(\bigcup_{z \in Z} \adh(z)) \cup \adh(x)$ where two vertices $u, v \in V(\adhconn(x))$ are adjacent if and only if they are connected in $G[\cone(x)] \setminus S$. We compute these along with some coloring that preserves connectivity information as follows.

\begin{plemma}\label{lemma:computing-all-adhconn}
    Given $\profile(x)$ and $\bgraph_S(x)$, one can compute the following for all $x \in Y$ in time $O(k^5)$.
    \begin{enumerate}
        \item The graph $\adhconn(x)$.
        \item A coloring $\col: V(\adhconn(x)) \rightarrow [O(k^2)]$ where $\col(u) = \col(v)$ if and only if $u, v \in V(\adhconn(x))$ are connected in $G[\cone(x)] \setminus S$.
        \item A inverse coloring $\invcol: [O(k^2)] \rightarrow \adh(x)$ such that $\invcol(i) = v$ where $v \in \adh(x)$ with $\col(v) = i$.
    \end{enumerate}
\end{plemma}

\proof We first consider computing $\adhconn(x)$ for a specific $x \in Y$. Suppose $Z \subseteq \children(x)$ is the set of its affected children. Note that $|Z| \leq k$ and all adhesions have at most $k = O(k)$ vertices, so $|V(\adhconn(x))| \leq (k + 1)k$. By its definition, we can use \cref{lemma:subset-connectivity-in-bgraphs} to compute $H = \adhconn(x)$, giving as input $\bgraph_S(x)$ and $V(H) = V(\adhconn(x))$, in $O(k^2 \cdot |V(\adhconn(x))|) \leq O(k^2 \cdot (k + 1)k) = O(k^4)$. Then, we perform BFS in $\adhconn(x)$ in $O(k^4)$ time to identify its connected components, and assign an arbitrary coloring $\col: V(\adhconn(x)) \rightarrow [O(k^2)]$ where only vertices in the same connected components in $\adhconn(x)$ have the same color. For each used color $i$, we assign $\invcol(i) = v$ by choosing an arbitrary $v \in \adh(x)$ with $\col(i) = v$, or assign $\invcol(i) = \emptyset$ if no such $v$ exist.

It is clear that computing $\adhconn(x)$ with its coloring $\col$ and inverse coloring $\invcol$ takes $O(k^4)$ time for each $x \in Y$. There are at most $O(k)$ important nodes, so computing these for all $x \in Y$ takes $O(k \cdot k^4) = O(k^5)$ time. This completes the proof.

\subsubsection{Precomputing connectivity within bag graphs}\label{sec:precompute-bgraph-connectivity}

To achieve fast query time, we precompute some vital connectivity information with regard to the patch set $\pset(x)$ of $\bgraph_S(x)$, for every node $x$ where $\profile(x)$ is computed. Note that $x$ can either be an important node or an affected children of some important node. For any such $x$, we say a patch $P \in \pset(x)$ is \textit{touched} if it either contains some vertex $v \in S$, or it contains both $u, v$ for some adhesion edge $(u, v)$ that is deleted during construction of $\bgraph_S(x)$. Let $A \subseteq \pset(x)$ be the set of all touched patches in $\pset(x)$. 

Suppose $Z$ is the set of affected children of $x$. For every affected child $z \in S$, the set of all $z$-supported adhesion edges form a clique in $\bgraph_{\emptyset}(x)$ on at most $k$ vertices and $k^2$ edges. Note that only at most $k/2 = O(k)$ of these edges are vertex disjoint, so at most $O(k)$ patches can be touched due to these $z$-supported adhesion edges. Thus we have $|A| \leq k + O(k \cdot k) = O(k^2)$ as $|S| \leq k$ and there can be at most $|Z| \cdot O(k) \leq O(k^2)$ disjoint patches touched due to adhesion edge deletions.

Additionally, for each small patch $P$, we need to update the set $N_{\adh}(P)$ of adhesion neighbors for each small patch $P \in \pset(x)$ due to the adhesion edge deletions in $\bgraph_S(x)$. We need to compute $A$ and update $N_{\adh}(P)$ for all small patches $P$ before proceeding to the next step.

\begin{pclaim}\label{claim:computing-affected-patches-and-updating-neighbor-sets}
    For a node $x \in V(T)$, given $\bgraph_S(x)$, one can in time $O(k^3)$ compute the set $A$ of affected patches and update the set $N_{\adh}(P)$ for each small patch $P \in \pset(x)$.
\end{pclaim}

\proof We have $\bgraph_S(x)$ and $\profile(z)$ for all $z \in Z$ where $Z$ is the set of its affected children since $x \in Y$. We first check all vertices $v \in S$. For any $v \in S \cap \bag(x)$, we identify the patch $P_v$ where $v \in P_v$ by looking at its stored pointer and include it in $A$. Then we compute the set of deleted adhesion edges by comparing for each $z \in Z$ the set of $z$-supported adhesion edges in $\bgraph_S(x)$ and $E(\profile(z))$. For any deleted adhesion edge $(u, v)$, we check if they are in the same patch by comparing their pointers, and we include $P$ if $u, v \in P$. It is clear that computing $A$ take $O(k^3)$ time as there are at most $O(k^3)$ deleted adhesion edges.

We now update $N_{\adh}(P)$ for each small path $P$. For any deleted adhesion edge $(u, v)$ in $\bgraph_S(x)$, we consider the patches $u$ and $v$ are in. Let $P_u, P_v$ be the patches containing $u, v$ respectively, we discard this adhesion edge if $P_u = P_v$. Otherwise, if $P_v$ is a small patch, we decrement the counter of $u$ if $u \in N_{\adh}(P_v)$, and delete it from $N_{\adh}(P_v)$ if the counter reaches zero. We do the same for $P_u$ if it is small. We can compute this in $O(k^3)$ time as there can be at most $O(k^3)$ deleted adhesion edges in $\bgraph_S(x)$. We then finish the proof.\\

Next, we precompute and store for each vertex in a touched patch whether it is in the unique connected component in $\bgraph_S(x)$ that has size more than $k$. If not, we label the vertex with the connected component it is in, formally stated in the following lemma.

\begin{plemma}\label{lemma:precompute-affected-patch-connectivity}
    For a node $x \in V(T)$, given $\bgraph_S(x)$, one can compute and label every vertex $v \in P$ for each patch $P \in A$ whether it is in the single connected component $C$ in $\bgraph_S(x)$ where $|C| > k$ in $O(k^5)$ time. Moreover, let $B$ be the set of all vertices reachable from $(\bigcup_{P \in A} P) \setminus C$ in $\bgraph_S(x)$, one can compute a coloring $\col: B \rightarrow [O(k)]$ such that $\col(u) = \col(v)$ for each $u, v \in B$ if and only if they are connected in $\bgraph_S(x)$ in the same time.
\end{plemma}

\proof Recall that each patch has size at most $k + 1 = O(k)$, then the total number of vertices in touched patches is $|\bigcup_{P \in A} P| \leq |A| \cdot O(k) = O(k^3)$. We first call \cref{lemma:subset-connectivity-in-bgraphs} with $\bgraph_S(x)$ and $V(H) = \bigcup_{P \in A} P$ to obtain $H$ in time $O(k^2 \cdot |H|) \leq O(k^2 \cdot k^3) = O(k^5)$. For each vertex $v \in V(H)$ that has been marked as in $C$ when computing $H$, we also label it as such in $\bgraph_S(x)$. 

Note that $|B| \leq k$ by \cref{lemma:bgraph-structure} since $B$ and $C$ are disjoint by definition. We can compute $B$ and its coloring $\col$ as follows. For each vertex $v \in V(H) \setminus C$, let $\col(v)$ be an arbitrary unused color in $[O(k)]$ if $v$ is not already colored, and perform BFS in $\bgraph_S(x)$ starting from $v$ to visit all vertices reachable from $v$. For each vertex $u$ reached from $v$, let $\col(u) = \col(v)$. Repeating this step for all vertices in $V(H) \setminus C$ takes $O(|B|^2) = O(k^2)$ time as by definition no vertex outside of $B$ can be reached. Thus the proof is complete.

\subsubsection{Complete update process}

Now we have finished every step in the update algorithm, let us summarize the complete update process below.

\begin{enumerate}
    \item Given the set of failed vertices $S$ where $|S| \leq k$, compute the set of important nodes $Y$ in $O(k^2)$ time.
    \item Compute $\profile(x)$ for all $x \in Y$ using \cref{lemma:computing-all-important-profiles}. The profiles of all affected children of every $x \in Y$ are also computed in this process.
    \item Compute $\adhconn(x)$ along with its coloring and inverse coloring for all $x \in Y$ using \cref{lemma:computing-all-adhconn}.
    \item Precompute bag graph connectivity information for every node $x$ where $\profile(x)$ is computed by repeating \cref{lemma:precompute-affected-patch-connectivity}.
\end{enumerate}

It is clear that the runtime of the update algorithm is dominated by step 4, where repeating \cref{lemma:precompute-affected-patch-connectivity} for all important nodes and their affected children takes $O(k^5) \cdot O(|Y|) \leq O(k^6)$ time. This proves the update time guarantee in \cref{thm:main}.

\subsection{Query phase}

Finally, we describe the algorithm this data structure uses to answer queries. Suppose we are given a query $u, v \in V(G)$ where we need to answer whether $u$ and $v$ are connected in $G \setminus S$. If either $u \in S$ or $v \in S$, we simply output "not connected", so let us assume that $u, v \notin S$ for the rest of this section.

We will repeat the same process for $u$ and $v$, then combine them at the end to answer the query, so for now let us focus on $u$. The goal is to lift $u$ up to the root and compute their connectivity in $\profile(r)$. The procedure is split into cases depending on the position of $u$, and we will start with the simple case of when $u$ is already included in some precomputed profile.

\subsubsection{Case 1: $u$ is in a profile computed during update}\label{sec:query-case1}

Let $x$ be a node where $u \in V(\profile(x))$ and $\profile(x)$ is computed in the update phase. Recall that $x$ is either an important node or an affected children of some important node. 

We \textit{augment} $u$ to some $\profile(x')$ by adding $u$ to $V(\profile(x'))$ and an edge $(u, u')$ to $E(\profile(x'))$ where $u'$ is a vertex in $\adh(x')$ and $u, u'$ are connected in $G[\cone(x')] \setminus S$. This generalizes the definition of profiles as the connected component containing $u$ in the augmented $\profile(x')$ may not be a clique. However, it remains true that every pair of vertices $u, v \in V(\profile(x'))$ are connected (rather than adjacent) in $\profile(x')$ if and only if they are connected in $G[\cone(x')] \setminus S$. We can augment $u$ to the profile of all important nodes along the path from $x$ to $r$ in $T$ quickly in this case.

\begin{plemma}\label{lemma:query-case1}
    Given the node $x$ where $u \in V(\profile(x))$, let $x_1, x_2,...,x_{\ell}, r \in Y$ be the sequence of important nodes along the $x$-$r$ path in $T$, one can in time $O(k)$ augment $u$ to $\profile(x_i)$ for all $1 \leq i \leq \ell$ and $\profile(r)$.
\end{plemma}

\proof If $x \notin Y$, then $x_1 = \parent(x)$ as $x$ must be its affected child. Otherwise $x_1 = x$. We can find the sequence $x_1,...,x_{\ell},r$ in $O(k)$ time by traversing the tree $T_Y$ constructed in the proof of \cref{lemma:computing-all-important-profiles} from node $x_1$ to its root, which must also be $r$ since $r \in Y$. Note that $\ell \leq O(k)$ as $|Y| \leq O(k)$.

Let $u_1 \in V(\profile(x_1))$ such that $u$ and $u_1$ are connected in $G[\cone(x_1)] \setminus S$. If $x = x_1$, we simply set $u_1 = u$. If $x \neq x_1$, we can find $u_1$ in constant time. Note that in this case $u \in V(\adhconn(x_1))$. Let $\col_{x_1}$ and $\invcol_{x_1}$ be the coloring and inverse coloring computed with $\adhconn(x_1)$ by \cref{lemma:computing-all-adhconn}, we simply let $u_1 = \invcol_{x_1}(\col_{x_1}(u))$. One can easily verify that $u_1$ is correct by the definition of these functions. Then, add $u$ to $V(\profile(x_1))$ and include the edge $(u, u_1)$ in $E(\profile(x_1))$.

Next, we move up to $x_2$. Let $z_2 = \dir(x_2, x_1)$ in $T$ and $z_2$ must be an affected child of $x_2$. We obtain $\torso(z_2, x_1)$ from \cref{lemma:torso-query} in constant time, and get $(\profile(z_2), \col_{z_2}, \invcol_{z_2})$ using \cref{lemma:combining-profiles-with-torsos} by giving $\profile(x_1)$ and $\torso(z_2, x_1)$ as input (removing $u$ from $\profile(x_1)$ for indexing if it is newly added). Let $u_2' = \invcol_{z_2}(\col_{z_2}(u_1))$, then by definition $u_2'$, $u_1$, and $u$ must be connected in $G[\cone(z_2)] \setminus S$. Note that $u_2' \in V(\adhconn(x_2))$. Let $\col_{x_2}$ and $\invcol_{x_2}$ be the coloring and inverse coloring of $\adhconn(x_2)$, and let $u_2 = \invcol_{x_2}(\col_{x_2}(u_2'))$, then $u_2$, $u_2'$, and $u$ must be connected in $G[\cone(x_2)] \setminus S$. We add $u$ to $V(\profile(x_2))$ and $(u, u_2)$ to $E(\profile(x_2))$. If $u_2' = \emptyset$, we additionally set $\col_{x_2}(u) = -1$ and $\invcol_{x_2}(-1) = u$. If $u_2' \neq \emptyset$ but $u_2 = \emptyset$, we set $\invcol_{x_2}(\col_{x_2}(u_2')) = u$. Note that every step takes $O(1)$ time.

We repeat the same process to move from $x_{i - 1}$ to $x_{i}$, , updating the coloring and inverse coloring of $\adhconn(x_i)$ along the process, until we reach $x_{\ell}$ and then $r$. Since $\ell \leq O(k)$ and each move-up step takes $O(1)$ time, augmenting $u$ to the profile of all nodes $x_1,...,x_{\ell}, r$ takes $O(k)$ time in total. Thus we have proved the lemma.\\

Let us move on to the next case. Suppose $x_u \in V(T)$ is the unique node where $u \in \mrg(x_u)$. We can find $x_u$ by looking at the pointer stored with $u$. In the following case we consider when $x_u$ is an important node.

\subsubsection{Case 2: $x_u \in Y$}\label{sec:query-case2}

Because $x_u$ is important, $\profile(x_u)$ is also computed during update. It must be true that $u \notin V(\profile(x_u))$ as $u \in \mrg(x_u)$. Thus, the goal of this case is to augment $u$ to the profile of $x_u$ quickly and return to Case 1.

\begin{plemma}\label{lemma:query-case2}
    Given $x_u \in Y$, one can in time $O(k)$ augment $u$ to $\profile(x_u)$.
\end{plemma}

\proof We aim to find whether or not $u$ is contained in the single connected component $C$ in $\bgraph_S(x)$ with $|C| > k$. If so, we find some $u' \in \adh(x_u)$ that has been marked as in $C$, and add $u$ along with the edge $(u, u')$ in $\profile(x_u)$.

Let $P_u$ be the patch in $\pset(x_u)$ where $u \in P_u$. This can be found in constant time by looking at the pointer stored along side $u$ in the processed bag graph of $x_u$. Suppose $A \subseteq \pset(x_u)$ is the set of affected patches. If $P_u \in A$, then recall \cref{lemma:precompute-affected-patch-connectivity}, we would have already computed whether $u$ is in $C$. If $u$ is not labeled so, we look up $\col(u)$ to find some $u'' \in \adh(x_u)$ where $\col(u'') = \col(u)$, and add $u$ and $(u, u'')$ to $\profile(x_u)$.

If $P_u \notin A$ and $P_u$ is a big patch, then it must be true that $u \in C$ since $P_u$ has $k + 1$ connected vertices and otherwise unbreakability is contradicted. Finally, if $P_u \notin A$ and $P_u$ is a small patch, we need to decide whether $P_u \subseteq C$. We check every $v \in N(P_u)$. The set $N(P)$ must contain a neighbor that is not failed if such neighbor exists since $|N(P_u)| = k + 1$ if $P_u$ has more than $k$ neighbors. If all vertices in $N(P_u)$ are failed, we find another arbitrary neighbor by checking $k + 1$ remaining vertices in $N_{\adh}(P_u)$. Similarly, one such vertex must be in the set if such a neighbor exists since $|N_{\adh}(P_u)| = k^3 + 1$ if there are more than $k^3$ adhesion neighbors. If there exist some neighbor $v$ that is not failed, from \cref{claim:pset-structure} we know that $v$ must be in a big patch $P_v$. If $P_v \notin A$, then $P_v, P_u \subseteq C$ and thus $u \in C$. Otherwise, we can check the label or color of $v$ and augment $u$ accordingly. If all vertices in $N(P_u)$ are failed and $E_{\adh}(P_u)$ is empty, we check if there is a vertex $u'' \in P_u$ such that $u'' \in \adh(x_u)$ and add $u$ with $(u, u'')$ to $\profile(x_u)$. If no such $u''$ exist, we only add $u$ to $\profile(x_u)$ without any edges.

Note that in each scenario at most $O(k)$ vertices are checked, so this step takes $O(k)$ time, finishing the proof.\\

With $u$ augmented to $\profile(x_u)$, we can return to Case 1 since $x_u \in Y$.

Now we consider the cases when $x_u \notin Y$.

\subsubsection{Case 3: $x_u \notin Y$ and $x_u$ is affected}\label{sec:query-case3}

Recall that $x_u$ is affected when $\comp(x_u) \cap S \neq \emptyset$. In this case, there must be exactly one affected child of $x_u$ as otherwise $x_u$ would be in the $\lca$ closure thus be in $Y$. Let this unique affected child be $z$. Recall from \cref{section:single-child-failure} that $x_u$ is in the case of single-child failure as $S$ is also disjoint from $\bag(x_u)$ because $x_u \notin Y$. We will exploit this fact when augmenting $u$ to $\profile(x_u)$.

\begin{pclaim}\label{claim:query-case3}
    Given $x_u \notin Y$ and $x_u$ is in the case of single-child failure, one can in time $O(k)$ augment $u$ to $\profile(x_u)$.
\end{pclaim}

\proof Let $y$ be the important node in the subtree $T(x_u)$ with the smallest depth. We can find $y$ in $O(k)$ time by checking all $y \in Y$ and select the node $y$ where $\lca(x_u, y) = x_u$ with the smallest depth. Then the single affected child of $x$ would be $z = \dir(x_u, y)$. First we obtain $\torso(x_u, y)$ and $\torso(z, y)$ by calling \cref{lemma:torso-query}, then get $\profile(x_u)$ and $\profile(z)$ by \cref{lemma:combining-profiles-with-torsos}. Note that we have $\profile(y)$ to do so because $y \in Y$.

Next, we apply \cref{lemma:single-child-failure-query}, given as input $x_u$, $\profile(z)$, and $u$ to obtain $L, I, \col$ as well as the label of $u$ in $O(k)$ time. Suppose $C$ is the connected component with $|C| > k$ in $\bgraph_S(x_u)$. If the label indicates that $u \in C$, we choose an arbitrary vertex $u' \in L$ and augment $u$ to $\profile(x_u)$ with the edge $(u, u')$. Otherwise, we choose some $u'' \in \adh(x)$ where $\col(u) = \col(u'')$ and augment $u$ to $\profile(x_u)$ via $(u, u'')$. Note that checking all vertices in $\adh(x)$ takes $O(k)$ time, thus we can augment $u$ to $\profile(x_u)$ in $O(k)$ time. This proves the lemma.\\

To return to Case 1, we still need to augment $u$ to a precomputed profile. Let $x_1$ be the first affected node on the $x_u$-$r$ path in $T$ and let $z_1 = \dir(x_1, x_u)$. Note that $z_1$ must be an affected child of $x_1$ since $z_1 = \dir(x_1, y)$ where $y$ is the important node in $T(x_u)$ with the smallest depth. Thus augmenting $u$ to the profile of $z_1$ would suffice as $\profile(z_1)$ is precomputed.

\begin{plemma}\label{lemma:query-case3}
    Given $x_u \notin Y$ and $x_u$ is in the case of single-child failure, suppose $x_1$ is the first important node along the $x_u$-$r$ path in $T$, one can augment $u$ to $\profile(z_1)$ where $z_1 = \dir(x_1, x_u)$ in $O(k)$ time.
\end{plemma}

\proof We start by using \cref{claim:query-case3} to augment $u$ to $\profile(x_u)$ in $O(k)$ time. Next we compute $x_1$ in time $O(k)$ by selecting the node $x_1 \in Y$ where $\lca(x_1, x_u) = x_1$ with the largest depth. Let $z_1 = \dir(x_1, x_u)$ and suppose $y$ is the smallest-depth important node in $T(x_u)$. There cannot be any important node in $V(T(z_1)) \setminus V(T(x_u))$ because otherwise there will be some node $x' \in Y$ along the $x_u$-$z_1$ path that is the lowest common ancestor of such an important node and $y$, contradicting the selection of $x_1$.

This implies that $\comp(z_1) \setminus \cone(x_u)$ is disjoint from $S$, and we can apply \cref{lemma:combining-profiles-with-torsos} to obtain $(\profile(z_1), \col, \invcol)$. Suppose $u'$ is a neighbor of $u$ in the augmented $\profile(x_u)$. Let $u'' = \invcol(\col(u'))$, the $u, u'$, and $u''$ must be connected in $\bgraph_S(z_1)$ by definition. This way we can augment $u$ to $\profile(z_1)$ with the edge $(u, u'')$. If $u'' = \emptyset$, we let $\col(u) = -1$ and $\invcol(-1) = u$. Thus we can augment $u$ to $\profile(z_1)$ in $O(k)$ time, finishing the proof.\\

We now return to Case 1 with $\profile(z_1)$ augmented with $u$.

The final remaining case is when $x_u \notin Y$ and $x_u$ is unaffected.

\subsubsection{Case 4: $x_u \notin Y$ and $x_u$ is unaffected}\label{sec:query-case4}

In this case, there must not be any node $y \in T(x_u)$ with $y \in Y$ because $\comp(x_u) \cap S = \emptyset$. Then, by the regularity of the decomposition $\mathcal{T}$, we know that $\profile(x_u)$ is simply a complete graph. Since $G[\cone(x_u)]$ is connected in this situation, we can augment $u$ by adding the edge $(u, u')$ in $\profile(x_u)$ where $u'$ is an arbitrary vertex in $\adh(x_u)$. We will next show how to augment $u$ to a precomputed profile.

\begin{plemma}\label{lemma:query-case4}
    Given $x_u \notin Y$ and that there are no important nodes in the subtree $T(x_u)$, one can in time $O(k)$ augment $u$ to precomputed profile.
\end{plemma}

\proof Suppose $x_1$ is the first important node on the $x_u$-$r$ path in $T$. We begin by computing $x_1$ in $O(k)$ time by selecting the node $x_1 \in Y$ with the largest depth where $\lca(x_1, x_u) = x_1$. Let $z_1 = \dir(x_1, x_u)$. There are two situations in this case, and let us first consider when $z_1$ is an affected child of $x_1$. In this situation, there must be exactly one node $z'$ such that $\parent(z')$ is a node in the $z_1$-$x_u$ path and the subtree $T(z')$ contains some important nodes. This is true because otherwise a node in the $z_1$-$x_u$ path will be a lowest common ancestor of two important nodes and be in $Y$, contradicting the selection of $x_1$. Let $x' = \parent(z')$ and let $z'' = \dir(x', x_u)$.

Since there must not be any important nodes in $T(z'')$, the subgraph $G[\cone(z'')]$ must be connected by regularity of $\mathcal{T}$. This implies that $u$ is connected to all vertices in $\adh(z_u')$, so we can find an arbitrary vertex $u' \in \adh(z'')$ as a proxy of $u$, to augment $u$ to $\profile(z_1)$ we simply augment $u'$ to $\profile(x_1)$ and replace $u'$ with $u$. Note that $u \in \bag(x')$ and $z'$ is the single affected child of $x'$, meaning that $x'$ is in the case of single-child failure, the we are back in Case 3 and can use \cref{lemma:query-case3} to augment $u'$ (and thus $u$) to $\profile(z_1)$ in $O(k)$ time.

Now let us consider the situation when $z_1$ is not affected. Then $G[\cone(z_1)]$ must be connected by regularity, and we can find an arbitrary vertex $u' \in \adh(z_1)$ as a proxy of $u$. Since $u' \in \bag(x_1)$ and $x_1 \in Y$, we return to Case 2 and can use \cref{lemma:query-case2} in $O(k)$ time to augment $u'$ to $\profile(x_1)$. Replacing $u'$ by $u$ in $\profile(x_1)$ would be a valid $u$-augmented $\profile(x_1)$ as $u$ and $u'$ are connected. Therefore, we conclude the proof as in both situations, we augment $u$ to a precomputed profile in $O(k)$.\\

We can then return to Case 1 since $u$ is augmented in a precomputed profile.

\subsubsection{Complete query process}

Given the query $u, v \in V(G)$, we first identify $x_u, x_v$ where $u \in \mrg(x_u), v \in \mrg(x_v)$ using their stored pointers. Starting with $u$, we first decide if $x_u \in Y$. If so, we proceed to Case 2. Otherwise, we decide between Case 3 and Case 4 by checking if there exists an important node $y \in Y$ where $\lca(x_u, y) = x_u$ in $O(k)$ time. We proceed to Case 3 if such $y$ exists as $x_u$ must be affected, and we move to Case 4 otherwise. The same process is repeated for $v$. After all steps are done, we check $\profile(r)$, which contains exactly $u, v$ since $\adh(r) = \emptyset$. One can easily verify by the proofs in each step that $\profile(r)$ is a valid profile, thus $u$ and $v$ are connected in $G[\cone(r)] \setminus S = G \setminus S$ if and only if they are connected in $\profile(r)$. We finally output "connected" if $u, v$ are connected in $\profile(r)$, and output "not connected" otherwise.

Since each step takes $O(k)$ time, one can answer the query on any $u, v \in V(G)$ in $O(k)$ time. This proves the update time guarantee of \cref{thm:main} and thus completes the proof of the entire theorem.

\section{Data structure for single-child failure}\label{sec:single-child-failure}
In this section we give the proof for \cref{lemma:single-child-failure-query}, which is the only missing component of proving \cref{thm:main}. The lemma is restated below for convenience.

\begin{plemma}
    There is a data structure that uses $2^{O(k^2)} n$ space, initializes in $2^{O(k^2)} n$ time, and answers the query in $O(k)$ time: given a node $x \in V(T)$ in the case of single-child failure, $\profile(z)$ with $z$ being its single affected child, and a vertex $u \in \bag(x)$, output the following.
    \begin{enumerate}
        \item Let $C$ be the unique connected component in $\bgraph_S(x)$ where $|C| > k$, the set $L = \adh(x) \cap C$.
        \item The set of vertices $I$ where each $v \in I$ is not in $C$ and is connected to some $v' \in \adh(x)$, with a coloring $\col: I \rightarrow [O(k))]$ such that $\col(u) = \col(v)$ for any $u, v \in I$ if and only if they are connected in $\bgraph_S(x)$.
        \item A label indicating whether $u \in C$.
    \end{enumerate}
\end{plemma}

\proof For a specific node $x$ in this case, given $\profile(z)$, we can build $\bgraph_S(x)$ explicitly. As $z$ is the single affected child of $x$, only the connectivity provided by the $z$-supported adhesion edges will be affected by the deletion of $S$ in $\bgraph_{\emptyset}(x)$. Other adhesion edges and edges in $G[\bag(x)]$ should remain unchanged as $S$ is disjoint from the components at all other children of $x$ and $\bag(x)$. Thus, to obtain $\bgraph_S(x)$ from $\bgraph_{\emptyset}(x)$, we only need to delete any $z$-supported adhesion edge $(u, v)$ that is not in $\profile(z)$ since $u, v$ would not be connected in $G[\cone(z)] \setminus S$ by the definition of profiles. This takes $O(k^2)$ time as there are at most $k^2$ adhesion edges supported by $z$ and checking whether a specific edge is in $\profile(z)$ takes constant time.

Then we can compute $L$ and $I$ given a specific $x$ and $\profile(z)$ using $\bgraph_S(x)$. For every vertex $v \in \adh(x)$, we perform BFS starting from $v$ for $k + 1$ steps. If more than $k$ vertices are reached, $v$ must be in $C$ by \cref{lemma:bgraph-structure}, and we add $v$ to $L$. Otherwise, we add all visited vertices to $I$ and assign these vertices with the same color in $\col$. This takes $O(k^3) = O(k^3)$ time in total since $\adh(x) \leq k$ and each $(k + 1)$-step BFS takes $O(k^2)$ time. Since both $L$ and $I$ have size at most $O(k)$, this takes $O(k)$ space to store.

Next, we process the patch set $\pset(x)$ with respect to $\bgraph_S(x)$. Using \cref{claim:computing-affected-patches-and-updating-neighbor-sets}, we can compute the set $A$ of affected patches and update the set $N_{\adh}(P)$ for each small patch $P$. Let $U_{\adh}$ be the set of updated outgoing adhesion edge sets, then it takes $O(k^3 \cdot k^3) = O(k^6) = O(k^6)$ space to store and time to compute since there are at most $k^3$ deleted adhesion edges and each edge set stores at most $k^3 + 1$ edges. Using the method described in \cref{lemma:precompute-affected-patch-connectivity}, we can compute and store the additional described labels and colorings, which take $O(k^5)$ time.

To achieve fast query time, we precompute and store $L, I, \col$ with $U_{\adh}$ and other labels and colorings for all possible $x$ and $\profile(z)$ queries by constructing the following. For each non-root node $z \in V(T)$, let $x = \parent(z)$, and we store a table $F_z$ indexed by all possible profiles. Let $S_p$ be the set of all graphs on $k$ vertices, which must contain all possible profiles since any profile has at most $k$ vertices. Note that $|S_p| = 2^{k^2} = 2^{O(k^2)}$. For each $G_p \in S_p$, we compute its entry $F_z(G_p)$ by setting $\profile(z) = G_t$, construct $\bgraph_S(x)$ by deleting edges in $\bgraph_{\emptyset}(x)$ accordingly, compute $L, I, \col, U_{\adh}$ with other labels and colorings from it, and store them in $F_z(G_p)$. When given any query on $x$ and $\profile(z)$, we can look up the entry $F_z(\profile(z))$ in constant time.

There are $|T| - 1 = O(n)$ such tables and each has $2^{O(k^2)}$ entries, where the content of each entry takes $O(k^6)$ space to store and $O(k^6)$ time to compute. Then we can prepare this table in $2^{O(k^2)} \cdot O(k^6n) = 2^{O(k^2)}n$ time and space.

Now for any query given $x$, $\profile(z)$, and any $u \in \bag(x)$, we can first return $L, I, \col$ in constant time by looking up the table. We then decide the label of $u$ using the process shown in \cref{lemma:query-case2}. Note that we can do this because we have stored all the updated neighbor sets with labels and colorings required by this procedure. Therefore, we can answer any query in $O(k)$ time. This completes the proof.

\section{Tradeoffs between update time and query time}\label{sec:tradeoffs}
In this section we present two different configurations of our update and query algorithms. More specifically, we will show how to obtain a connectivity oracle that updates in $O(k^5)$ time and answers query in $O(k^2)$ time, as well as an oracle with $O(k^4)$ update time and $O(k^3)$ query time. We will not build additional structures or do any additional precomputing during the preprocessing and update phases respectively.

The first configuration of our data structure is summarized in the following theorem.

\begin{ptheorem}\label{thm:config1}
    There exists a deterministic vertex-failure connectivity oracle that uses $2^{O(k^2)}n + O(k^2n \alpha_c(n))$ space, $k^{O(k^2)} n + O(m + k^3 n \log^2 n + k^6n \log n)$ preprocessing time, $O(k^5)$ update time, and $O(k^2)$ query time.
\end{ptheorem}

\proof The data structure and preprocessing steps remain the same as described in \cref{sec:prep}. During the update step, given the set $S$ of at most $k$ failed vertices, we compute all important profiles and adhesion connectivity graphs for all important nodes as described in \cref{sec:computing-important-profiles} and \cref{sec:computing-adhconn} respectively. However, we skip the steps in \cref{sec:precompute-bgraph-connectivity} that updates connectivity in affected patches for the bag graphs of every important node. Now, the runtime of the update algorithm is dominated by \cref{lemma:computing-all-adhconn}, which is $O(k^5)$, proving the update time guarantee.

Now we consider the query phase. Note that for Case 1 in \cref{sec:query-case1}, only $S$-restricted bag graphs, precomputed profiles, and adhesion connectivity graphs along with their colorings are used, so we have computed everything required by this case during the new update procedure. Same goes for Case 3 and 4 in sections \cref{sec:query-case3}, \cref{sec:query-case4} respectively, as the only additional information they require is the auxiliary single-child failure structure of \cref{lemma:single-child-failure-query}. Thus the runtime remains to be $O(k)$ for these three cases. For Case 2 in \cref{sec:query-case2}, however, we do not have updated bag graphs for important nodes. Instead, we simply perform BFS starting from $u$ for $k + 1$ steps. If it terminates before reaching $k + 1$ vertices, we check all visited vertices and choose some visited $u' \in \adh(x_u)$ to augment $u$ to $\profile(x_u)$ with the edge $(u, u')$. Otherwise, $u$ must be in the unique large connected component $C$ in $\bgraph_S(x)$, and we augment $u$ via edge $(u, u'')$ where $u'' \in \adh(x_u)$ is an arbitrary vertex marked as in $C$. This step takes $O(k^2)$ time, so we can still handle Case 2 in $O(k^2)$ time, proving the update time guarantee. This finishes the proof.\\

Next, we describe the second configuration of our connectivity oracle that avoids the $k$-exponential factor is the space usage, presented in the following theorem.

\begin{ptheorem}\label{thm:config2}
    There exists a deterministic vertex-failure connectivity oracle that uses $O(k^2 n \alpha_c(n))$ space, $k^{O(k^2)} n + O(m + k^3 n \log^2 n + k^6n \log n)$ preprocessing time, $O(k^4)$ update time, and $O(k^3)$ query time.
\end{ptheorem}

\proof In the preprocessing phase, we avoid initializing and storing the profile computing structure given by \cref{lemma:combining-profiles-with-torsos} as well as the structure for single-child failure of \cref{lemma:single-child-failure-query}. We also modify the torso query structure given by \cref{lemma:torso-query} in the following way. We still first compute all $\torso(x, y)$ by \cref{claim:computing-all-torsos} and compute a shortcutting of $T$ with $O(n \alpha_c(n))$ edges with \cref{thm:tree-shortcutting}. Then, for each edge $(x, y)$ in the shortcutting where $x$ is an ancestor of $y$, we store $\torso(x, y)$ explicitly. Each torso has $O(k)$ vertices thus takes $O(k^2)$ to store, thus the total space usage of this structure is $O(k^2 n \alpha_{c/2 + 1}(n))$. When $\torso(x, y)$ is queried, let $y, x_1, x_2,...,x_h, x$ be the shortcutting $y$-$x$ path, where $h \leq c$  is a constant. We construct the graph $H = \torso(y, x_1) \cup \torso(x_1,x_2) \cup ... \cup \torso(x_h, x)$ and for each $v \in \adh(x) \cup \adh(y)$, we perform a BFS starting from $v$ and terminating a branch once it reaches some $u \in \adh(x) \cup \adh(y)$. For all visited $u \in \adh(x) \cup \adh(y)$, we add the edge $(u, v)$ to $\torso(x, y)$. Note that this construction is correct by definition. Since $|E(H)| = O(k^2)$, each BFS takes $O(k^2)$ time. This is repeated for $O(k)$ times since $|\adh(x) \cup \adh(y)| = O(k)$, thus answering the query $\torso(x, y)$ takes $O(k^3)$ time.

For the update algorithm, we will only computed all important profiles as show in \cref{sec:computing-important-profiles}, and the steps in \cref{sec:computing-adhconn} and \cref{sec:precompute-bgraph-connectivity} will both be skipped. We only need to modify \cref{claim:building-affected-child-profile} in the following way. We still first find the topmost important node $y$ in the subtree $T(z)$ and query $\torso(z, y)$ in time $O(k^3)$. Then, we perform BFS on the graph $\profile(y) \cup \torso(z, y)$, and add an edge $(u, v)$ to $\profile(z)$ if $u$ and $z$ are in the same connected component for all $u, v \in \adh(z)$, taking $O(k^2)$ time. By the definition of profiles, this construction is correct. Note that in \cref{lemma:computing-all-important-profiles} this process is called $O(k)$ times, so the update time is still $O(k^4)$.

When answering query on vertices $u, v$, we first consider $u$ and repeat the same for $v$. Let $x_u \in V(T)$ be the unique node where $u \in \mrg(x_u)$. If $\profile(x_u)$ is not computed during update, we first compute it in time $O(k^3)$ the same way by finding the topmost important node $y$ in $T(x_u)$ and combining $\torso(x_u, y)$ and $\profile(y)$, and augment $u$ to it via a $(k + 1)$-step BFS in $O(k^2)$ time. Then, let $x_1,...,x_h$ be all the important nodes in the $x_u$-$r$ path, which can be computed in $O(k^2)$ time using $\lca$ queries and depth. We then augment $u$ to $\profile(x_1)$ by querying $\torso(x_1, x_u)$ in $O(k^3)$ time, and performing BFS in $\torso(x_1, x_u) \cup \profile(x_1)$ from a neighbor of $u$ in $\profile(x_u)$ in time $O(k^2)$ to find a proxy of $u$ in $\profile(x_1)$. We repeat the same process for augmenting $u$ to $z_2 = \dir(x_2, x_1)$, only now $\torso(z_2, x_1)$ would be available to use since it has been queried in the update phase so only $O(k^2)$ time is required. Then, we augment $u$ to $\profile(x_2)$ by finding its proxy via performing a $(k + 1)$-step BFS in $\bgraph_S(x_2)$ starting from the proxy of $u$ in $\profile(z_2)$. This is repeated for $O(k)$ times until $u$ is augmented to $\profile(r)$, taking $O(k^3)$ time. Therefore, in total, we can answer queries in $O(k^3)$ time, completing the proof.\\

Note that for all three configurations in \cref{thm:main}, \cref{thm:config1}, and \cref{thm:config2}, the underlying data structure and preprocessing steps are all the same. Thus after building the connectivity oracle, at every update step, one can select which update and query configuration is optimal depending on the number of queries that follow it.

\section{Vertex-Cut Oracles}
\label{sec:cutoracle}
In this section, we discuss two problems closely related to the vertex-failure connectivity oracle problem: the \textit{vertex-cut oracle} problem and its generalization, the \textit{Steiner vertex-cut oracle} problem. 

A $k$-vertex-cut oracle for an undirected graph $G$ is a data structure that supports vertex-cut queries: given a set $F$ of vertices with $|F|\leq k$, does $F$ form a global vertex cut, or more generally, how many connected components does $G\setminus F$ have? The more general Steiner $k$-vertex-cut oracles will further receive a set $A$ of terminal vertices, and the corresponding Steiner-vertex-cut queries ask whether $A\setminus F$ is disconnected in $G\setminus F$, or more generally, the number of connected components in $G\setminus F$ intersecting $A$.

%In this section, we will discuss a reduction from the vertex cut oracle problem to the vertex-failure connectivity oracle problem, which is formally stated in \Cref{thm:VertexCutOracle}.

In what follows, we introduce simple reductions from these two problems to the vertex-failure connectivity oracle problem. These two reductions both work for the harder queries which ask for the number of connected components.

\begin{ptheorem}[Refinement of \cite{DBLP:conf/icalp/Kosinas25}]
\label{thm:VertexCutOracle}
Let $G$ be an undirected graph with a given parameter $k$. Suppose there is a $k$-vertex-failure connectivity oracle of $G$ with preprocessing time $\sfP$, space $\sfS$, update time $\sfU$ and query time $\sfQ$. Then there exists a $k$-vertex cut oracle with preprocessing time $\sfP + O(m + kn\log n)$, space $\sfS + O(kn)$ and query time $\sfU + |F|\cdot \sfQ + O(2^{|F|}\cdot k^{2})$.

If the input graph $G$ is further known to be $k$-vertex-connected, the query time of the above $k$-vertex cut oracle can be improved to $\sfU + k\cdot \sfQ + O(k^{2})$.
\end{ptheorem}

The reduction from vertex-cut oracles to vertex-failure connectivity oracles (\Cref{thm:VertexCutOracle}) is implicitly given by \cite{DBLP:conf/icalp/Kosinas25}. Most interestingly, after slight refinement using dictionary data structures (e.g., \cite{hagerup2001deterministic}), this reduction is actually nearly lossless: all overhead terms are subsumed by the known lower bounds for vertex-failure connectivity oracles and vertex-cut oracles (except for a $\log n$ factor in the preprocessing time\footnote{It can be removed using the classic FKS perfect hashing \cite{fredman1984storing} if randomization is allowed.}). Therefore, it will give new vertex-cut oracles with query time purely depending on $k$, by plugging in such vertex-failure connectivity oracles (e.g., our oracles, \cite{pilipczuk2021algorithms}). For completeness, we include the proof of \Cref{thm:VertexCutOracle} in \Cref{sec:appendix1}.

%This reduction is a slight refinement of the DFS-based vertex cut oracle in \cite{DBLP:conf/icalp/Kosinas25}, but we want to highlight its conceptual significance here. By plugging in the vertex-failure connectivity oracles from our \Cref{thm:main} and \cite{long2024better}, we obtain the following tradeoffs of vertex cut oracles.

%\begin{pcorollary}
%Let $G$ be an undirected graph with a given parameter $k$. There is a deterministic $k$-vertex cut oracles admitting the following tradeoffs:
%\begin{enumerate}
%\item preprocessing time $k^{O(k^{2})}n + O(m) + \tilde{O}(k^{6}n)$, space $O(k^{2}n\alpha_{c}(n))$ and query time $O(k^{4} + 2^{k}\cdot k^{2})$, or
%\item preprocessing time $O(m + kn^{1+o(1)})$, space $\tilde{O}(kn)$ and query time $\tilde{O}(k^{2}) + O(2^{k}\cdot k^{2})$.
%\end{enumerate}
%If the input graph $G$ is further known to be $k$-vertex-connected, then the query time in the above tradeoffs can be improved to $O(k^{4})$ and $\tilde{O}(k^{2})$ respectively.
%\end{pcorollary}

Moreover, we show that the above reduction generalizes to the Steiner vertex-cut oracle problem, as shown in \Cref{thm:SteinerVertexCutOracle}. The proof of \Cref{thm:SteinerVertexCutOracle} is deferred to \Cref{sec:appendix2}.

\begin{ptheorem}
\label{thm:SteinerVertexCutOracle}
Let $G$ be an undirected graph with a given terminal set $A\subseteq V(G)$ and a given parameter $k$. Suppose there is a $k$-vertex-failure connectivity oracle of $G$ with preprocessing time $\sfP$, space $\sfS$, update time $\sfU$ and query time $\sfQ$. Then there exists a $k$-Steiner vertex cut oracle with preprocessing time $\sfP + O(m + kn\log n)$, space $\sfS + O(kn)$ and query time $\sfU + |F|\cdot \sfQ + O(2^{|F|}\cdot k^{2} + k^{2}\cdot\log n)$.

If the terminal set $A$ is further known to be $k$-vertex-connected in $G$, the query time of the above $k$-Steiner vertex cut oracle can be improved to $\sfU + |F|\cdot \sfQ + O(k^{2}\cdot \log n)$.
\end{ptheorem}

\section{Discussions and Open Problems}\label{sec:discussions}
The ideal bounds for a $k$-vertex-failure connectivity oracle are
\[
    O(\min\{m,nk\}) \text{ space}, \qquad
    O(m) \text{ preprocessing}, \qquad
    O(k^2) \text{ update}, \qquad
    O(k) \text{ query}.
\]
Long et~al.~\cite{long2024,long2024better}
achieve all four bounds simultaneously up to $\polylog n$ factors.
Moreover, each of these bounds is separately conditionally optimal under
standard hypotheses~\cite{HKNS15,long2024}. Thus, the main remaining
frontier is to remove the dependence on $n$ from the update and query times
while keeping near-linear preprocessing and space.

Our oracle makes progress on this frontier. It has near-linear preprocessing,
uses
\(
    O(k^2 n \alpha_c(n))
\)
space, and supports updates and queries in small $\poly(k)$ time. Currently,
the two known approaches to $n$-independent update and query times are the
\emph{algebraic approach} of \cite{van2019sensitive} and the
\emph{unbreakable-decomposition approach} of
\cite{pilipczuk2021algorithms}. The following challenges highlight what
seems to require new techniques.

\begin{enumerate}[leftmargin=*]
    \item \textbf{Near-linear space and preprocessing without exponential
    $k$-dependence.}
    The algebraic approach uses $\Omega(n^2)$ space, while current
    unbreakable-decomposition based oracles incur exponential dependence on
    $k$ in the preprocessing. Can one obtain
    \[
        \widetilde O(n \poly(k)) \text{ space}
        \quad\text{and}\quad
        \widetilde O(m \poly(k)) \text{ preprocessing}
    \]
    together with $\poly(k)$ update and query time?

    \item \textbf{Optimal $k$-dependence for updates and queries.}
    All known $n$-independent oracles have update time at least
    $k^\omega$ or a larger polynomial in $k$. Can one achieve the
    conditionally optimal update time $O(k^2)$? More ambitiously, can this
    be achieved simultaneously with the conditionally optimal query time
    $O(k)$?

    \item \textbf{Exactly linear space in $n$.}
    Our shortcutting scheme reduces the space to
    $O_k(n\alpha_c(n))$, but still leaves the tiny $\alpha_c(n)$ factor.
    This factor appears tied to the constant-hop tree shortcutting used by
    our method; see the shortcutting lower bounds in
    \cite{bhattacharyya2012transitive}. Pilipczuk et
    al.~\cite{pilipczuk2021algorithms} achieve $O_k(n)$ space, but with
    large polynomial preprocessing time and $2^{2^{O(k)}}$ update and query
    time. Can one obtain
    \[
        O_k(n) \text{ space}, \qquad
        \widetilde O_k(m) \text{ preprocessing}, \qquad
        \poly(k) \text{ update and query time}?
    \]
    Such a result would likely require a way to bypass our shortcut-based
    technique.
\end{enumerate}

\appendix

\section{Omitted Proofs}
\subsection{Proof of \Cref{thm:VertexCutOracle}}
\label{sec:appendix1}

Before describing the oracle, we introduce some notations related to DFS trees. We fix an arbitrary DFS tree $T$ in $G$ by starting a DFS from an arbitrary root vertex. Note that $T$ is rooted, and thus induces a natural ancestor-descendant relationship on its vertices. For each non-root vertex $v$, its parent is denoted by $\parr(v)$. We define a natural \emph{depth} function $\depth$ on vertices by setting $\depth(r) = 1$ for the root $r$ and $\depth(v) = \depth(\parr(v))+1$ for each non-root vertex $v$.

A well-known property of the DFS tree $T$ is that, for each edge $e=(u,v)\in E(G)$, either $u$ or $v$ (called the \emph{lower endpoint} of $e$) must be an ancestor of the other in $T$. This property motivates the following notation. For each vertex $v$, we define the back-edges of $v$ to be
\[
\backedge(v) = \{(u,v)\in E(G)\mid \text{$u$ is an ancestor of $v$ in $T$}\}.
\]
Let $T[v]$ denote the subtree rooted at $v$. Then naturally we can define the back-edges of $T[v]$ to be
\[
\backedge(T[v]) = \{(u,w)\in E(G)\mid w\in T[v]\text{ and $u$ is an ancestor of $v$}\}.
\]
Let $\backvertex(T[v])$ be a vertex set that collects the lower endpoints of all back-edges in $\backedge(T[v])$.

\paragraph{The Oracle.} The $k$-vertex cut oracle consists of the following.
\begin{enumerate}
\item\label{Item:CutOracle1} The $k$-vertex-failure connectivity oracle with update time $\sfU$ and query time $\sfQ$.
\item\label{Item:CutOracle2} The DFS tree $T$, and a data structure that supports the following level-ancestor queries on $T$ in $O(1)$ time: given a vertex $v$ and a parameter $d<\depth(v)$, which vertex is the ancestor of $v$ with depth $d$?
\item\label{Item:CutOracle3} For each vertex $v$, an ordered list ${\cal L}(v)$ consisting of the $k+1$ vertices in $\backvertex(T[v])$ with the smallest depths (or all vertices in $\backvertex(T[v])$ if there are fewer than $k+1$ of them), sorted in increasing order of depth.
%a boolean value that indicates whether $\backvertex(T[v])$ has at least $k+1$ vertices or not. If not, store the set $\backvertex(T[v])$ explicitly as an ordered list, with vertices sorted by increasing depth.
\item\label{Item:CutOracle4} For each vertex $v$, a dictionary data structure that supports the following membership and counting queries in $O(k)$ time: given a set $F'\subseteq V$ with at most $k$ vertices (sorted by increasing depth), how many children $v'$ of $v$ have the ordered list $\backvertex(T[v])$ exactly the same as $F'$?

\end{enumerate}
Below, we analyze the preprocessing time and space of this oracle. 

First of all, we can always assume the input graph has only $\bar{m}:=O(nk)$ edges by doing the Nagamochi-Ibaraki sparsification \cite{nagamochi1992computing} in advance, which takes additional $O(m)$ preprocessing time. \Cref{Item:CutOracle1} takes $\sfP$ preprocessing time and $\sfS$ space. In \Cref{Item:CutOracle2}, the DFS tree takes $O(\bar{m})$ time to compute and the level-ancestor data structure takes $O(n)$ preprocessing time and space by \cite{bender2004level}. 

\Cref{Item:CutOracle3} can be computed in $O(nk)$ time as follows. 
%Note that it suffices to compute, for each vertex $v$, an ordered list ${\cal L}(v)$ consisting of the $k+1$ vertices in $\backvertex(T[v])$ with the smallest depths (or all vertices in $\backvertex(T[v])$ if there are fewer than $k+1$ of them), sorted in increasing order of depth. To do this, 
We first compute auxiliary lists ${\cal L}'(v)$ defined w.r.t. $\backvertex(v)$ for all vertices $v$. All these ${\cal L}'(v)$ can be computed in $O(\bar{m})$ time with bucket sort. Now, we are ready to compute all lists ${\cal L}(v)$ in a bottom-up manner. 
\begin{itemize}
\item For each leaf vertex, ${\cal L}(v)$ is exactly ${\cal L}'(v)$.
\item For each non-leaf vertex $v$, the key observation is that ${\cal L}(v)$ consists exactly of the vertices of depth less than $\depth(v)$ appearing in ${\cal L}'(v)$ and in the ${\cal L}$-lists of $v$’s children (truncating to the first $k+1$ such vertices if there are more). Therefore, we can enumerate the vertices in ${\cal L}(v)$ one by one as follows: at each step, we linearly scan the heads of ${\cal L}'(v)$ and the ${\cal L}$-lists of all children, select the minimum-depth vertex, and then pop it from every list in which it appears.
\end{itemize}
Therefore, computing \Cref{Item:CutOracle3} takes $O(nk)$ preprocessing time and space in total.

\begin{plemma}[\cite{hagerup2001deterministic}]
\label{lemma:DetDictionary}
Consider a word-RAM machine with word size $w$. Given a set $S$ in a universe $U=\{0,1\}^{w}$, where each element in $S$ can be associated with satellite data, there exists a deterministic data structure that, on a query $x\in U$, determines whether $x\in S$ and, if so, returns the satellite data associated with $x$. The data structure has preprocessing time $O(|S|\log |S|)$, space $O(|S|)$ and query time $O(1)$.
\end{plemma}

For \Cref{Item:CutOracle4}, we can exploit the Hagerup-Miltersen-Pagh deterministic dictionaries (see \Cref{lemma:DetDictionary}). Note that the data structures in \Cref{Item:CutOracle4} are essentially dictionaries with those explicitly stored lists $\backvertex(T[v])$ as keys. Note that each such key consists of $k$ integers and therefore does not fit in a single word, but this can be handled using standard techniques. For example, we can build a trie over these keys and store, at each node, a standard dictionary in \Cref{lemma:DetDictionary} to support child lookup. This results in a dictionary of $n'$ many $k$-integer keys with preprocessing time $O(n'k\log n')$, space $O(n')$ and query time $O(k)$. Applying this modified dictionary, the data structures in \Cref{Item:CutOracle4} in total take preprocessing time $O(nk\log n)$, space $O(n)$ and query time $O(k)$.

Summing over all the items, we obtain the desired bounds on preprocessing time and space.

\paragraph{The Query Algorithm.} Consider a query with a given $F\subseteq V$ of at most $k$ vertices. We want to determine the number of connected components in $G\setminus F$.

Let ${\cal T}$ be the collection of subtrees in $T\setminus F$. For each subtree $\tau\in {\cal T}$, let $r_{\tau}$ denote its root. Following \cite{kosinas2023}, we classify the subtrees into \emph{internal subtree} and \emph{hanging subtree}. A subtree $\tau\in {\cal T}$ is a hanging subtree if $\tau$ is exactly $T[r_{\tau}]$ (equivalently, there is no $F$-vertex in $T[r_{\tau}]$), and all remaining subtrees are internal. For better understanding, we point out that there are at most $k$ internal subtrees while the number of hanging subtrees can be very large. The answer to this query is given by the following observation.

\begin{plemma}[\cite{kosinas2023}]
\label{lemma:CutQueryCorrectness}
The number of connected components in $G\setminus F$ is the summation of the following.
\begin{enumerate}[label=(\alph*)]
\item\label{Item:CutQuery1} The number of connected components in $G\setminus F$ which contain at least one internal subtree. 
\item\label{Item:CutQuery2} The number of hanging subtrees $\tau$ such that $\backvertex(T[r_{\tau}])\subseteq F$.
\end{enumerate}
\end{plemma}
The intuition behind \Cref{lemma:CutQueryCorrectness} is that each hanging subtree $\tau$ either forms a connected component on its own when $\backvertex(T[r_{\tau}]) \subseteq F$, or else it must connect to some internal subtree through a back edge. In what follows, we discuss how to compute \Cref{Item:CutQuery1} and \Cref{Item:CutQuery2} respectively.

\medskip

\noindent{\underline{Computing \Cref{Item:CutQuery1}.}} We first pick a set $X$ of vertices
\[
X := \{\parr(v)\mid v\in F\}\setminus F.
\]
Intuitively, $X$ hits and only hits all internal subtrees: we can easily observe that $X$ is a subset of the union of internal subtrees, and $X$ contains at least one vertex from each internal subtree. Therefore, \Cref{Item:CutQuery1} now equals the number of connected components in $G\setminus F$ which contain at least one $X$-vertex. The latter can be computed by choosing an arbitrary vertex $x\in X$ and then querying the vertex-failure connectivity oracle (\Cref{Item:CutOracle1}) for the connectivity between $x$ and every other $X$-vertex in $G\setminus F$. This involves one update operation to \Cref{Item:CutOracle1} and $|X|-1 = k-1$ query operations to \Cref{Item:CutOracle1}. Hence, computing \Cref{Item:CutQuery1} takes $O(|F|) + \sfU + (|F|-1)\cdot \sfQ$.

\medskip

\noindent{\underline{Computing \Cref{Item:CutQuery2}.}} Observe that each hanging subtree $\tau$ must have $\parr(r_{\tau})\in F$. Therefore, we can consider the children of $F$-vertices (called \emph{$F$-children}) instead. Note that not every child of an $F$-vertex is the root of a hanging subtree, and we call such children \emph{bad $F$-children} (in fact, a bad $F$-child is either an $F$-vertex or the root of an internal subtree). Now, \Cref{Item:CutQuery2} equals (i) the number of $F$-children $v$ such that $\backvertex(v)\subseteq F$, minus (ii) the number of bad $F$-children $v$ such that $\backvertex(v)\subseteq F$. 
\begin{itemize}
\item To compute (i), we can enumerate all $2^{|F|}$ subsets $F'$ of $F$, and then query the dictionary (\Cref{Item:CutOracle4}) of each $F$-vertex with the subset $F'$. This takes $O(2^{|F|}\cdot k^{2})$ time.
\item To compute (ii), we can first get all bad $F$-children explicitly in $O(k^{2})$ time, by performing level-ancestor queries between every pair of $F$-vertices, because a vertex $w$ is a bad $F$-child if and only if there exists two $F$-vertices $u,v$ such that $w$ is a children of $u$ and an ancestor of $v$. Then, we just check, for each bad $F$-child $v$, whether $\backvertex(v)\subseteq F$ or not (using \Cref{Item:CutOracle3}). The number of bad $F$-children is bounded by $k$ (since each $F$-vertex can contribute at most one bad child to its nearest strict ancestor in $F$), and checking each of them takes $O(k)$ time, so the second step takes $O(k^{2})$ time as well.

\end{itemize}
By the analysis above, the query time is bounded by $\sfU + |F|\cdot\sfQ + O(2^{|F|}\cdot k^{2})$.

\paragraph{Improvements for $k$-Vertex-Connected Graphs.} Consider a query $F$ with at least two hanging subtrees (the corner case where there is only one hanging subtree can be answered similar to the computation of \Cref{Item:CutQuery1} above). The key observation is that every hanging subtree $\tau$ must have at least $k$ vertices in $\backvertex(T[r_{\tau}])$, otherwise the input graph is not $k$-vertex-connected (since removing $\backvertex(T[r_{\tau}])$ will at least disconnect two hanging subtrees). This means \Cref{Item:CutQuery2} is now exactly the number of hanging subtrees $\tau$ with $\backvertex(T[r_{\tau}])=F$, and we only need to consider $F'=F$ (instead of $2^{|F|}$ subsets) when computing (i) of \Cref{Item:CutQuery2}. The query time is thus $\sfU + k\cdot \sfQ + O(k^{2})$.

\subsection{Proof of \Cref{thm:SteinerVertexCutOracle}}
\label{sec:appendix2}

To generalize the above results to the Steiner connectivity setting, we focus on describing the query algorithm for a query $F$, and mention the changes to the oracle along the way. First of all, we have the following generalized version of \Cref{lemma:CutQueryCorrectness}.

\begin{plemma}
\label{lemma:SteinerCutQueryCorrectness}
The number of connected components in $G\setminus F$ intersecting the terminals $A$ is the summation of the following.
\begin{enumerate}[label=(\Alph*)]
\item\label{Item:SteinerCutQuery1} The number of connected components in $G\setminus F$ which contain at least one internal subtree and intersect $A$. 
\item\label{Item:SteinerCutQuery2} The number of hanging subtrees $\tau$ such that $\backvertex(T[r_{\tau}])\subseteq F$ and $\tau$ intersects $A$.
\end{enumerate}
\end{plemma}

\noindent\underline{Computing \Cref{Item:SteinerCutQuery2}.} This part is almost identical, except that we only focus on subtrees intersecting $A$. Concretely, we change \Cref{Item:CutOracle4} to the following Item 4', and store \Cref{Item:CutOracle5}.
\begin{itemize}
\item[4'.] For each vertex $v$, a dictionary data structure that answers: given a set $F'\subseteq V$ with at most $k$ vertices, how many children $v'$ of $v$ have $\backvertex(T[v])=F'$ \emph{and $T[v]$ intersects $A$}? 
\end{itemize}
\begin{enumerate}[start=5]
\item\label{Item:CutOracle5} For each vertex $v$, the number of $A$-vertices in $T[v]$, denoted by $\alpha(T[v])$.
\end{enumerate}
With these changes, the algorithm computing \Cref{Item:SteinerCutQuery2} is straightforward, which also takes $O(2^{|F|}\cdot k^{2})$ time, and we omit the description here.

\medskip

\noindent\underline{Computing \Cref{Item:SteinerCutQuery1}.} This part needs more modification. Recall that the above algorithm for \Cref{Item:CutQuery2} actually returns, for each connected component of $G\setminus F$ that contains at least one internal subtree (called an \emph{internal component}), the set of internal subtrees it contains. Thus, it suffices to check, for each internal component, whether it intersects $A$ or not. Note that we can obtain the number of $A$-vertices inside an internal subtree (using \Cref{Item:CutOracle5}). However, this is not sufficient because such a component may contain hanging subtrees in addition to internal subtrees. Therefore, we need to extend the oracle with the following. 
\begin{itemize}
\item[2'.] In addition to \Cref{Item:CutOracle2}, further store, for each vertex $v$, its DFS number (i.e., $v$'s preorder number in the DFS tree $T$), denoted by $\dfs(v)$, as well as the minimum and maximum DFS numbers among the vertices in its subtree, denoted by $\dfs_{\min}(T[v])$ (in fact, $\dfs_{\min}(T[v])$ equals $\dfs(v)$) and $\dfs_{\max}(T[v])$.
\end{itemize}
\begin{enumerate}[start=6]
\item\label{Item:CutOracle6} For each vertex $v$, a \emph{weighted} ordered list ${\cal R}(v)$ that stores all elements (with multiplicity) in the ${\cal L}$-lists of all $v$'s children, sorted in increasing order of their DFS numbers. An element in ${\cal R}(v)$ originally from the list ${\cal L}(v')$ is \emph{weighted} by $\alpha(T[v'])$ (i.e., the number of $A$-vertices in $T[v']$). Lastly, store the prefix sum of the weights of elements in ${\cal R}(v)$.
\end{enumerate}

Next, we discuss how to check in the query phase whether an internal component intersects $A$ or not. To this end, we compute, for each internal subtree $\tau$, the number of $A$-vertices inside $\tau$, denoted by $\lambda(\tau)$, as well as the following value
\[
\kappa(\tau):= \sum_{v\in \tau}\ \sum_{\substack{\text{hanging subtree $\tau'$}\\\text{s.t. }v\in {\cal L}(r_{\tau'})}} \alpha(T[r_{\tau'}]).
\]
We point out that a hanging subtree $\tau'$ may have its $\alpha(T[r_{\tau'}])$ counted multiple times in $\kappa(\tau)$, if there are multiple vertices $v$ in the internal subtree $\tau$ satisfying $v\in {\cal L}(r_{\tau'})$, so $\kappa(\tau)$ does not admit an intuitive interpretation. But we can still use these $\kappa(\tau)$ to check whether an internal component intersects $A$ or not, by the following \Cref{claim:CheckInternalComponent}.
\begin{pclaim}
\label{claim:CheckInternalComponent}
An internal component intersects $A$ if and only if
\[
\lambda(\tau) + \sum_{\substack{\textnormal{internal subtree $\tau$}\\\textnormal{in this component}}} \kappa(\tau)\geq 1.
\]
\end{pclaim}
\begin{proof}
The ``if'' statement is trivial. To see the ``only if'' statement, consider an internal component with an $A$-vertex $u$ inside. If $u$ belongs to an internal subtree, then $\lambda(\tau)\geq 1$ must hold and we are done. So suppose $u$ belongs to a hanging subtree $\tau'$, which implies $\alpha(T[r_{\tau'}])\geq 1$. Note that this hanging subtree cannot have ${\cal L}(r_{\tau'})\subseteq F$ (otherwise it will form a non-internal component by itself), so pick a vertex $v\in {\cal L}(r_{\tau'})\setminus F$ and $v$ must belong to an internal subtree $\tau$ in this component. Finally, note that $\alpha(T[r_{\tau'}])$ contributes to $\kappa(\tau)$, and thus the LHS of the inequality is positive.
\end{proof}

Therefore, it remains to compute $\lambda(\tau)$ and $\kappa(\tau)$ for all internal subtrees $\tau$. Together with the time needed for \Cref{Item:CutQuery1}, this determines the overall running time of \Cref{Item:SteinerCutQuery1}. 
\begin{itemize}
\item Computing $\lambda(\tau)$ for all internal subtrees can be done in $O(k^{2})$ or even $O(k\log k)$ time using \Cref{Item:CutOracle5} and the level-ancestor data structure \Cref{Item:CutOracle2}. It is a simple exercise and we omit the proof here.
\item To compute $\kappa(\tau)$, recall our discussion from computing \Cref{Item:CutQuery2}: the hanging subtrees are exactly those subtrees rooted at non-bad $F$-children. Hence, the high-level idea is still computing 
\[
\kappa'(\tau):=\sum_{v\in\tau}\sum_{\substack{\text{$F$-children $v'$}\\\text{s.t. }v\in {\cal L}(v')}}\alpha(T[v']),
\]
and then deducting the contributions from bad $F$-children. To compute $\kappa'(\tau)$, the standard way is to exploit the interval-representation of $\tau$ on the DFS order, and the contribution of an interval can be easily obtained by doing binary search on the ${\cal R}$-lists of $F$-vertices. In summary, the computational time of all $\kappa(\tau)$ is $O(k^{2}\log n)$, where the $\log n$ factor comes from the binary search.
\end{itemize}
Therefore, the total computational time for \Cref{Item:SteinerCutQuery1} is $\sfU + |F|\cdot Q + O(k^{2}\log n)$.

\bibliographystyle{alpha}
\bibliography{ref}

\end{document}